\documentclass[aps,prx,twocolumn,amsmath,nofootinbib,longbibliography,amssymb,superscriptaddress,10pt]{revtex4-2}
\usepackage[margin=1in]{geometry}
\usepackage[english]{babel}
\usepackage{etoolbox}
\usepackage[shortlabels]{enumitem}
\usepackage[scr=boondox]{mathalpha}

\usepackage[margin=1in]{geometry} 
\usepackage{graphicx}
\usepackage[caption=false]{subfig}
\graphicspath{
    {figures/}
}

\usepackage{mathtools} 

\usepackage[english]{babel}
\usepackage{amsmath,amsthm,amssymb,amsfonts}
\usepackage[colorlinks=true, pdfstartview=FitV, linkcolor=blue, citecolor=blue, urlcolor=blue]{hyperref}
\usepackage[shortlabels]{enumitem}
\usepackage{bm}
\usepackage{dsfont} 
\usepackage{etoolbox}

\docsvlist{A,B,C,D,E,F,G,H,I,J,K,L,M,N,O,P,Q,R,S,T,U,V,W,X,Y,Z}

\docsvlist{A,B,C,D,E,F,G,H,I,J,K,L,M,N,O,P,Q,R,S,T,U,V,W,X,Y,Z}

\docsvlist{A,B,C,D,E,F,G,H,I,J,K,L,M,N,O,P,Q,R,S,T,U,V,W,X,Y,Z}

\docsvlist{d}

\docsvlist{E,P,R,e,k,p,q,r,g,x,y}

\usepackage{ stmaryrd }
\def\llb{\llbracket}
\def\rrb{\rrbracket}

\usepackage{tikz-cd} 
\usetikzlibrary{matrix,arrows,decorations.pathmorphing}
\usetikzlibrary{fit, shapes.geometric,calc}

\def\bra{\langle}
\def\ket{\rangle}

\def\backslash{\symbol{92}}
\def\ve{\varepsilon}

\DeclareSymbolFont{symbolsC}{U}{txsyc}{m}{n}
\SetSymbolFont{symbolsC}{bold}{U}{txsyc}{bx}{n}
\DeclareFontSubstitution{U}{txsyc}{m}{n}
\DeclareMathSymbol{\multimap}{\mathrel}{symbolsC}{"15}

\def\Tor{\mathrm{Tor}}
\def\sq{\mathrm{sq}}
\def\hon{\mathrm{hon}}
\def\tri{\mathrm{tri}}
\def\hor{\mathrm{h}}
\def\ver{\mathrm{v}}
\DeclareMathOperator{\cone}{cone}
\DeclareMathOperator{\im}{im}

\newcommand{\norm}[1]{\left\lVert#1\right\rVert}
\newcommand{\dnum}[1][1]{\mathds{#1}}

\DeclareMathOperator{\coker}{coker}

\DeclareMathOperator{\supp}{supp}

\newtheorem{theorem}{Theorem}[section]

\newtheorem{prop}[theorem]{Proposition}
\newtheorem{lemma}[theorem]{Lemma}
\newtheorem{OQ}{Open Question}

\theoremstyle{definition}
\newtheorem{definition}[theorem]{Definition}
\newtheorem{example}{Example}[section]
\newtheorem{remark}{Remark}

\newcommand{\beq}{\begin{equation}}
\newcommand{\eeq}{\end{equation}}

\usepackage{xcolor}
\newcommand{\red}[1]{{\color{red} #1}}
\newcommand{\gray}[1]{{\color{gray} #1}}
\newcommand{\blue}[1]{{\color{blue} #1}}

\begin{document}
\title{Unified Framework for Quantum Code Embedding}

\author{Andrew C. Yuan}
\affiliation{Condensed Matter Theory Center and Joint Quantum Institute, Department of Physics, University of Maryland, College Park, Maryland 20742, USA}

\begin{abstract}
    Given a Calderbank-Shor-Steane (CSS) code, it is sometimes necessary to modify the code by adding an arbitrary number of physical qubits and parity checks.
    Motivations may include concatenating codes, embedding low-density parity check (LDPC) codes into finite-dimensional Euclidean space, or reducing the weights of parity checks.
    During this embedding, it is essential that the modified code possesses an isomorphic set of logical qubits as the original code. 
    However, despite numerous explicit constructions, the conditions of when such a property holds true is not known in general.
    Therefore, using the language of homological algebra, we provide a unified framework that guarantees a natural isomorphism between the output and input codes.
    In particular, we explicitly show how previous works fit into our framework.
\end{abstract}
\maketitle

\section{Introduction}

The physical realization of a quantum computer is generally plagued by problems such as decoherence and systematic errors in realizing quantum gates.
Quantum error correction (QEC) aims to resolve this issue by utilizing multiple physical qubits in the representation of a single logical qubit.
Stabilizer codes are of particular interest in QEC due to their simple construction using only parity checks, i.e., tensor products of Pauli operators over qubits. 
Codes that involve a small number of physical qubits, e.g., the 9-bit Shor code and the 7-bit Steane code, have been shown to encode a single logical qubit protected against any single qubit error, and thus can be utilized in fault-tolerant quantum computing by concatenating codes \cite{aharonov1997fault,shor1996fault,knill1996concatenated,knill1996threshold}.

Alternatively, codes that involve a large number of qubits can be fault-tolerant without further manipulation provided that they possess low-density parity checks (LDPC) \cite{gottesman2013fault,kovalev2012fault}.
Specifically, the \textit{local} structure of LDPC codes guarantees that if the error rate is $O(1)$-small so that error clusters do not percolate, then the disconnected error clusters can be corrected. 
An important subclass of LDPC codes is \textit{topological} codes, since their error-correcting capability stems from the topology of the underlying manifold.
The seminal 2D toric code \cite{kitaev2003fault,dennis2002topological,raussendorf2007fault}, for instance, is constructed on a discretized torus with $n$ qubits placed on edges of the lattice.
The existence of non-contractible loops on the torus then guarantees that the code distance scales as $d=\Theta(\sqrt{n})$.
The toric code, however, is far from optimal, as it encodes only a constant number of logical qubits $k=\Theta(1)$.
Only recently has there been a significant breakthrough in finding \textit{good} LDPC codes -- those for which the code parameters $k,d$ scale linearly in $n$ \cite{panteleev2022asymptotically,leverrier2022quantum,dinur2023good}.

In the study of stabilizers, it is sometimes necessary to modify an input code by adding ancillas in a certain manner to obtain an output code with further desirable properties.
In such cases, despite being technically distinct codes, the logical subspaces are preserved in some natural manner and thus may be viewed as \textit{equivalent} codes.
One of the most famous examples would be code concatenation \cite{aharonov1997fault,shor1996fault,knill1996concatenated,knill1996threshold}, in which the physical qubits of code $A$ are replaced by the logical qubits of another code $B$\footnote{$B$ is usually chosen to be $n$ copies of $A$ where $n$ is the number of physical qubits in $A$} to obtain the output code $C$.  
From an operator perspective, this is equivalent to substituting the physical Paulis $Z_i,X_i$ of code $A$ with corresponding nontrivial logical representations $\bar{Z}_i,\bar{X}_i$ of code $B$, where the former, latter indices $i$ denote physical, logical qubits of code $A, B$, respectively.
During this embedding $A\mapsto C$, the logicals are thus expected to be naturally preserved.

Another famous example would be topological codes defined on the same manifold but constructed from different discretizations.
Since the logicals are related to the (singular) homology of the manifold, these codes are supposed to exhibit naturally isomorphic logicals. 
While the intuition is straightforward, a rigorous proof requires relatively in depth knowledge of singular, simplicial, and CW homology (see, e.g., Chap. 2 of \cite{hatcher2005algebraic}).
The issue becomes more subtle when boundary conditions are considered.
A notable example is the surface code variant of the toric code \cite{bravyi1998quantum, kitaev2012models}. 
Although such a code supports logical string operators reminiscent of that on a torus, it cannot be realized by any 2D CW complex\footnote{The $0$-homology of the surface code is 0, while any CW complex on a connected manifold must have rank 1 $0$-homology}, and thus whether similar techniques in homology can be applied becomes less apparent.

For more general codes, the question becomes increasingly relevant.
In recent year, explicit embeddings which preserve the logical qubits have been proposed for various application purposes, often in an attempt to modify the exotic good LDPC codes into more practical forms.
Hastings, for example, proposed a general method for reducing weights in quantum codes \cite{hastings2016weight,hastings2021quantum}.
Although LDPC codes are defined by having parity checks with $O(1)$ weight, in practice, this constant can still be large; therefore, reducing the check weight can help lower measurement errors.
This research also connects closely to fault-tolerant logical measurements \cite{williamson2024low,ide2025fault,cross2024improved,horsman2012surface}. 
Specifically, if one wishes to measuring a logical operator $\ell^\star$, direct measurement is not feasible nor fault-tolerant due to $\ell^\star$'s large weight ($\ge d$ code distance), making indirect measurement through a finite number of local operations preferable. 
This concept is equivalent to artificially adding $\ell^\star$ as a parity check to the original code, and subsequently reducing the weight of $\ell^\star$.

Another line of work involves embedding LDPC codes in $D$-dimensional Euclidean space \cite{williamson2024layer,portnoy2023local,lin2023geometrically}.
From a practical standpoint, LDPC codes are local in the sense of an exotic adjacency graph \cite{gottesman2013fault}, posing challenges for physical implementation.
This has motivated efforts to realize embeddings in Euclidean space, especially in low dimensions $D=3$.
Additionally, any local code in Euclidean space is subject to the Bravyi–Poulin–Terhal (BPT) bounds \cite{bravyi2009no,bravyi2010tradeoffs}, which constrain the code parameters as $k = O(n^{(D-2)/D})$ and $d=O(n^{(D-1)/D})$.
For instance, the toric code saturates the BPT bounds in $D=2$, but not in higher dimensions.
Hence, whether the BPT bounds can be saturated is also of theoretical interest.
This has been confirmed possible by combining recent discoveries of good LDPC codes \cite{panteleev2022asymptotically,leverrier2022quantum,dinur2023good} with optimal embeddings in $D \ge 3$ dimensions \cite{williamson2024layer,portnoy2023local,lin2023geometrically}.

In any case, the existing embeddings rely on explicit constructions and thus the general conditions sufficient to guarantee a natural isomorphism between logicals are unknown.
The main result of our manuscript is thus to address this issue by providing a unified framework, which, in particular, applies to all existing examples.
Our focus will be on Calderbank-Shor-Steane (CSS) codes, as any $[n,k,d]$ stabilizer code can be locally mapped to a $[4n,2k,2d]$ CSS code \cite{bravyi2010majorana}.
In the subsequent sections, we demonstrate how various known constructions naturally fit within this framework. 
As a byproduct, many of the proofs found in prior works are rendered more explicit and conceptually streamlined\footnote{
For example, the framework provides a straightforward Cleaning Lemma \eqref{lem:cleaning} which can be used to show that the embedding preserves the code distance, and unifies the proofs of Theorem \eqref{thm:layer-code-distance} \cite{williamson2024layer}, Theorem \eqref{thm:square-distance} \cite{lin2023geometrically} and Proposition \eqref{prop:weight-cone-distance} \cite{hastings2021quantum,williamson2024low}.
}.
Since we claim the framework to be natural, we discuss how the result can be regarded as a generalization of code concatenation -- Remark \eqref{rem:main}, and can be utilized for topological codes directly, without knowledge of the point-set topology of the underlying manifold -- Section \eqref{sec:topology}.
This generality also allows the inclusion of boundary conditions within the same unified treatment.
Since the upcoming sections primarily concern applications, we have structured the manuscript such that Sections \eqref{sec:topology}-\eqref{sec:weight} may be read independently, assuming familiarity with the preliminaries outlined in Section \eqref{sec:prelim} and the main result -- Theorem \eqref{thm:height-2}-\eqref{thm:height-n} and Cleaning Lemma \eqref{lem:cleaning}.




\subsection{High-Level Overview}

The framework, in its simplest form, can be regarded as a generalization of code concatenation.
As previously discussed, code concatenation replaces the physical qubits/Paulis of code $A$ with the logical qubits/operators of code $B$ to obtain output code $C$.
Note that CSS codes can be understood as $\dF_2$ chain complexes (see Section \eqref{sec:prelim}). 
For example, code $A$ can be written as $A_2 \to A_1 \to A_0$ where $A_2,A_1,A_0$ denotes the $Z$-type operators, qubits, and $X$-type operators, respectively, and the equivalence classes of $Z$- and $X$- type logicals correspond to the first homology $H_1(A)$ and cohomology $H^1(A)$, respectively.
The embedding procedure $A\mapsto C$ can be thus abstractly regarded as requiring\footnote{Since $H_1(B)\cong H^1(B)$, we omit the cohomology requirement.} 
\begin{equation}
    \label{eq:code-conc-requirement}
    A_1 = H_1(B)
\end{equation}

In the context of code concatenation, the naive embedding $A\mapsto C$ is sufficient since code $A$ is small in number of qubits and $B$ is just copies of $A$, one for each physical qubit\footnote{
Each hierarchy of concatenation results in the physical error rate to $p\mapsto p^2$, while the parity check weights $w\mapsto Cw$ where $C$ is a constant, so that $\sim \log \log (1/\ve)$ levels are necessary to reach an error rate of $\ve \ll 1$ and thus the parity check weights are at most logarithmic $\sim \log (1/\ve)$.
}.
In general, however, code $B$ could have a large code distance and thus the logicals $\bar{Z},\bar{X}$ must be large in weight.
This implies that even if $A$ is LDPC, the naive embedding does not guarantee that $C$ is still LDPC despite having the logical subspace preserved.
This is undesirable since LDPC is an essential property that guarantees fault-tolerance for large codes \cite{gottesman2013fault,quintavalle2021single} and thus a generalization is necessary.

Our framework thus generalizes the requirement in Eq. \eqref{eq:code-conc-requirement} by requiring that
\begin{equation}
    \label{eq:main-requirement}
    A_i = H_i(C^i), \quad i=2,1,0
\end{equation}
Where $C^i$ are independent CSS codes. 

A high-level way to understand the reasoning behind the generalization is by considering the following example.
As previously discussed, the naive requirement $A_1=H_1(C^1)$ may result in large check weights corresponding to $A_2,A_0$.
Consider a connected graph with vertices $\sV$, edges $\sE$ and generating cycles/faces $\sF$ so that the adjacency relation induces a complex $G=V\to E\to F$. 
Although the graph/complex can be local/LDPC, the homology $H_2(G)$ corresponds to the connected component with large weight $|\sV|$ and thus the requirement that $A_2 =H_2(G)$ (and similarly for $A_0$) aims to replace the large weight parity checks with local graphs so that the overall output code $C$ is LDPC.

\subsection{Main Results}


\begin{theorem}[Height-2 Cone]
    \label{thm:height-2} Let 
    \begin{equation}
    C=C_{2} \xrightarrow{\partial} C_{1} \xrightarrow{\partial} C_{0}
    \end{equation}
    be a $\dF_2$ chain complex such that each degree is the direct sum of $\dF_2$ vector spaces $C_{i} = C^{2}_{i} \oplus C^{1}_{i} \oplus C^{0}_{i}$, and the differential $\partial:C_{i} \to C_{i-1}$ is lower-triangular with respect to the decomposition, with matrix elements denoted as
    \begin{equation}
    \partial = 
    \begin{pmatrix}
        \partial^{2} &  &\\
        g_2 & \partial^{1} & \\
        p & g_1 & \partial^{0}
    \end{pmatrix}
    \end{equation}
    Then the following are $\dF_2$ chain complexes
    \begin{align}
        \label{eq:row-complex}
        C^{s} &= C^{s}_{2} \xrightarrow{\partial^s} C^{s}_{1} \xrightarrow{\partial^s} C^{s}_{0}, \quad s=2,1,0 \\
        \label{eq:column-complex}
        C^{\bg} &=H_2(\partial^{2}) \xrightarrow{[g_2]} H_1(\partial^{1}) \xrightarrow{[g_1]} H_0(\partial^{0})
    \end{align}
    Where $H_i(\partial^{s}),H_i(\partial^{\bg})$ are the homologies of $C^{s},C^{\bg}$, respectively, and  $[g_2],[g_1]$ are induced\footnote{
    The notation follows naturally if we write $\ell^{s}\mapsto [\ell^{s}]$ as the quotient map from $C^{s}_{i} \to H(\partial^{s})_{i}$.
    Specifically, the notation permits the concise relation $[g_2][\ell^2]=[g_2 \ell^2]$. 
    Equivalently, mapping an equivalence class $[\ell^2]\mapsto [g_2][\ell^2]$ is equal to first choosing any representation $\ell^2$, map it under the original map $\mapsto g_2 \ell^2$, and then obtain the equivalence class $\mapsto [g_2\ell^{2}]$.
    The case is similar for $[g_1]$.
    } from $g_2,g_1$ via the quotient maps.
    Refer to $C^{\bg}$ as the \textbf{embedded} complex within $C$ and $C$ as the \textbf{(height-2) cone} with \textbf{levels} $C^{s},s=2,1,0$.
    Moreover, if $H_1(\partial^{0}) =H_1(\partial^{2}) =0$ -- which we refer to as $C$ being \textbf{regular} (with respect to degree $1$) -- then $H_1 (\partial^{\bg}) \to H_1 (\partial)$ are isomorphic via the following map
    \begin{equation}
        \label{eq:height-2-cone-iso}
        \llb \ell^{1}_{1} \rrb \mapsto \left[
        \begin{pmatrix} 
        0 \\ \ell^{1}_{1} \\ \ell^{0}_{1}
        \end{pmatrix} \right]
    \end{equation}
    where $\ell^{1}_{1} \in C^{1}_{1}$ is a representation of $[\ell^{1}_{1}]\in H_1(\partial^{1})$, which is a further representation of $\llb \ell^{1}_{1}\rrb \in H_1(\partial^{\bg})$, and $\ell^{0}_{1} \in C^{0}_{1}$ is such that $g_1 \ell^{1}_{1}+\partial^{0} \ell^{0}_{1} =0$ (whose existence is guaranteed), and the right-hand-side is the equivalence class of the given element in $C_{1}$.

\end{theorem}

\begin{proof}
    See Appendix \eqref{proof:height-2}.
\end{proof}
\begin{remark}
\label{rem:main}
Let us first understand the statement (see also Remark \eqref{rem:necessity-of-generalized-cones}) and how it fits the high-level overview.
Since the matrix elements are $\dF_2$-linear maps from $C^{s}_{i} \to C^{r}_{i-1}$, we rewrite the $\partial$ as the following diagram, so that the diagonals represent the decomposition of $C_{i}$ and arrows denote the linear maps
\begin{equation}
\label{eq:height-2-diagram}
\begin{tikzpicture}[baseline]
\matrix(a)[matrix of math nodes, nodes in empty cells, nodes={minimum size=25pt},
row sep=2em, column sep=2em,
text height=1.25ex, text depth=0.25ex]
{&& \red{C^{2}_{2}}  & \red{C^{2}_{1}} & \red{C^{2}_{0}}\\
& \gray{C^{1}_{2}}  & \gray{C^{1}_{1}}  & \gray{C^{1}_{0}} &\\
\blue{C^{0}_{2}} & \blue{C^{0}_{1}} & \blue{C^{0}_{0}} &&\\};
\path[->,red,font=\scriptsize]
(a-1-3) edge node[above]{$\partial^{2}$} (a-1-4)
(a-1-4) edge node[above]{$\partial^{2}$} (a-1-5);
\path[->,gray,font=\scriptsize]
(a-2-2) edge node[above]{$\partial^{1}$} (a-2-3)
(a-2-3) edge node[above]{$\partial^{1}$} (a-2-4);
\path[->,blue,font=\scriptsize]
(a-3-1) edge node[above]{$\partial^{0}$} (a-3-2)
(a-3-2) edge node[above]{$\partial^{0}$} (a-3-3);
\path[->,font=\scriptsize]
(a-1-3) edge[bend right=65, dashed] node[left]{$p$} (a-3-2)
(a-1-4) edge[bend left=65, dashed] node[right]{$p$} (a-3-3);
\path[->,font=\scriptsize]
(a-1-3) edge node[right]{$g_2$} (a-2-3)
(a-1-4) edge node[right]{$g_2$} (a-2-4)
(a-2-2) edge node[right]{$g_1$} (a-3-2)
(a-2-3) edge node[right]{$g_1$} (a-3-3);
\draw[dashed] ($(a-1-3.north west) +(-0.1,0)$) rectangle ($(a-3-3.south east) +(0.1,0)$); 
\draw[dashed, opacity=0.3]
  ($(a-1-4.north west) + (.1,0)$) -- 
  ($(a-1-4.north east) + (-0.05,0)$) -- 
  ($(a-1-4.south east) +(-0.05,0.1)$) -- 
  ($(a-3-2.south east) +(-0.1,0)$) -- 
  ($(a-3-2.south west) +(0.05,0)$) -- 
  ($(a-3-2.north west) +(0.05,-0.1)$) -- cycle;
\end{tikzpicture}
\end{equation}
    
Since $\partial\partial=0$, it's straightforward to check that each row is a complex, i.e., $\partial^{a}\partial^{a}=0$, and that the squares are commutative, i.e., $g_2 \partial^{2} = \partial^{1} g_2$ and $g_1 \partial^{1} = \partial^{0} g_1$.
This implies that $g_2,g_1$ are chain maps and induces  maps $[g_2],[g_1]$ on the homologies.
The condition $\partial \partial =0$ also implies that $g_1 g_2$ is chain-homotopic to 0, i.e., $g_1 g_2 =  \partial^{0} p+ p\partial^{2}$ so that $p$ is a chain homotopy.
This then implies that $[g_1] [g_2] =0$ and thus the embedded column sequence $C^{\bg}$ is a well-defined chain complex.

If the embedded column $C^{\bg}$ and the rows $C^{s}$ are regarded as CSS codes, then the $Z$-type operators, qubits, and $X$-type operators of $C^{\bg}$ correspond to the homologies of $C^2,C^1,C^0$, respectively, as required by Eq. \eqref{eq:main-requirement}.
The theorem then essentially claims that if codes $C^{2},C^{0}$, which induce the parity checks of code $C^{\bg}$, do not have internal logical operators, i.e., $H_1(\partial^{2}) = H_1(\partial^{0})=0$, then the embedded code $C^{\bg}$ and code $C$ have isomorphic logical operators.
\end{remark}

\begin{remark}
\label{rem:necessity-of-generalized-cones}
We note that our main result is a generalization of the conventional (height-1) mapping cone \cite{weibel1994introduction}. 
In fact, it is an iterative application.
Specifically, if we regard $C_{i}^{2} \oplus C_{i}^{1}$ (or $C_{i}^{1} \oplus C_{i}^{0}$) as a single $\dF_2$-space for each $i$, the decomposition of $\partial$ is a $2\times 2$ block-triangular matrix and thus a conventional cone. 
Interesting, despite this straightforward generalization, only the sub-diagonal matrix elements (i.e., $g_2,g_1$) appear in the embedded code $C^{\bg}$ under general conditions.
This fact continues to hold true even for general height-$n$ cones as demonstrated in Theorem \eqref{thm:height-n}.

In terms of application, height-1 cones were first employed in Hastings' quantum weight reduction \cite{hastings2021quantum} and later identified as a framework for logical operator measurement \cite{ide2025fault}. 
In hindsight, the sufficiency of height-1 cones is somewhat intuitive in the latter context, since the removal of a logical operator is a \textit{classical} operation and classical codes can be understood as complexes of length 1 -- see Example \eqref{ex:classical}.
However, for general \textit{quantum} embeddings, a generalization to height-$(\ge 2)$ is necessary.
\end{remark}

As remarked, the main result can be straightforwardly generalized to height-$n$ cones\footnote{Further generalization to $R$-modules is straightforward. This may prove useful in considering translationally invariant CSS codes \cite{haah2013commuting} where $R$ is the Laurent polynomial ring or codes over general $q$-dimensional qudits so that $R =\dZ_q$, though it then becomes necessary to keep track of signs $\pm 1$ in front of matrix elements of $\partial$.}.

\begin{theorem}[Height-$n$ Cone]
    \label{thm:height-n} Let
    \begin{equation}
        C =\cdots \to C_{i} \xrightarrow{\partial} C_{i-1} \to \cdots
    \end{equation}
    be a chain complex over $\dF_2$ such that each degree $C_{i}$ is the direct sum $C_{i} = C^{n}_{i} \oplus C^{n-1}_{i} \oplus \cdots \oplus C^{0}_{i}$, and the differential $\partial$ is lower-triangular with matrix elements
    \begin{equation}
    \partial = 
    \begin{pmatrix}
        \partial^{n} &  &&&\\
        g_n & \partial^{n-1} &&& \\
        \cdot & g_{n-1} & \partial^{n-2} && \\
        \cdot & \cdot &\cdot & \ddots \\
        \cdot & \cdot & \cdot & g_{1} & \partial^{0}
    \end{pmatrix}
    \end{equation}
    where only the diagonal and sub-diagonal elements are explicitly shown.
    Then the following are $\dF_2$ chain complexes
    \begin{align}
        \label{eq:row-complex-R}
        C^{s} &= \cdots \to C^{s}_{i} \xrightarrow{\partial^{s}} C^{s}_{i-1} \to \cdots, \quad s=0,...,n\\
        \label{eq:column-complex-R}
        C^{\bg} &=H_n(\partial^{n}) \xrightarrow{[g_n]} H_{n-1}(\partial^{n-1}) \to \cdots \xrightarrow{[g_{1}]} H_{0}(\partial^{0})
    \end{align}
    Where $H_i(\partial^{s}),H_i(\partial^{\bg})$ are the homologies of $C^{s},C^{\bg}$, respectively, and $[g_s]$ are induced from $g_s$ via the quotient maps.
    We refer to $C^{\bg}$ as the \textbf{embedded} code within $C$ and $C$ as the \textbf{(height-n) cone} of $C^{\bg}$ consisting of \textbf{levels} $C^{s}$.
    
    Moreover, given fixed $0\le m\le n$, if $H_i(\partial^{s}) =0$ for all $0 \le s < i\le m$ and $m\le i < s\le n$, or equivalently, all spaces explicitly shown in the following diagram (irrelevant $C^{s}_{i}$ and arrows are omitted) except for the column sequence has trivial homology -- a property we refer to as being \textbf{regular} with respect to degree $m$,
    \begin{equation}
    \label{eq:height-n-diagram}
    \begin{tikzpicture}[baseline]
    \matrix(a)[matrix of math nodes, nodes in empty cells, nodes={minimum size=5mm},
    row sep=1.5em, column sep=2em,
    text height=1.5ex, text depth=0.25ex]
    { & & C^{n}_{n}  & \cdot & C^{n}_{0} & \\
      & & \cdot  & \cdot & \\
      & & C^{m}_{m} &  & \\
      & \cdot & \cdot &&\\
    C^{0}_{m} & \cdot & C^{0}_{0} &&\\
    };
    \path[->,font=\scriptsize]
    (a-1-3) edge node[above]{$\partial^{n}$} (a-1-4)
    (a-1-4) edge node[above]{$\partial^{n}$} (a-1-5)
    (a-2-3) edge (a-2-4)
    (a-4-2) edge (a-4-3)
    (a-5-1) edge node[above]{$\partial^{0}$} (a-5-2)
    (a-5-2) edge node[above]{$\partial^{0}$} (a-5-3)
    (a-1-3) edge node[right]{$g_n$} (a-2-3)
    (a-1-4) edge (a-2-4)
    (a-2-3) edge node[right]{$g_{m+1}$} (a-3-3)
    (a-3-3) edge node[right]{$g_{m}$} (a-4-3)
    (a-4-2) edge  (a-5-2)
    (a-4-3) edge node[right]{$g_{1}$} (a-5-3);
    \draw[dashed] (a-1-3.north west) rectangle (a-5-3.south east); 
    \draw[dashed, opacity=0.4]
      ($(a-1-5.north west) + (.1,0)$) -- 
      ($(a-1-5.north east) + (-0.05,0)$) -- 
      ($(a-1-5.south east) +(-0.05,0.1)$) -- 
      ($(a-5-1.south east) +(-0.1,0)$) -- 
      ($(a-5-1.south west) +(0.05,0)$) -- 
      ($(a-5-1.north west) +(0.05,-0.1)$) -- cycle;
    \end{tikzpicture}
    \end{equation}
    then $H_m (\partial^{\bg}) \mapsto H_m (\partial)$ are isomorphic via $\llb \ell^{m}_{m} \rrb \mapsto [\ell^{m}_{m} \oplus \ell^{m-1}_{m} \oplus \cdots \oplus \ell^{0}_{m}]$ where $\ell^{m}_{m}$ is a representation of $[\ell^{m}_{m}]\in H_m(\partial^{m})$, which is a further representation of $\llb \ell^{m}_{m}\rrb \in H_0(\partial^{\bg})$, and $\ell^{s}_{i}\in C^{s}_{i}$ are chosen so that $\ell^{m}_{m} \oplus \cdots \oplus \ell^{0}_{m} \in \ker \partial^{\bg}$ (whose existence is guaranteed).
    If the cone is regular with respect to all degrees $0,...,n$, then we say that the cone is \textbf{regular}.
\end{theorem}

\begin{proof}
    See Appendix \eqref{proof:height-n}.
\end{proof}
\begin{remark}
    Note that if $C_m^{m-1}=\cdots =C^0_{m}=0$ (so that Diagram \eqref{eq:height-n-diagram} is \textit{halved}\footnote{See, for example, Diagram \eqref{eq:barycenter-diagram}}), then the isomorphism is simplified to inclusion $\llb \ell_m^m\rrb \mapsto [\ell_m^m]$.
\end{remark}

In addition to a natural isomorphism, the framework also provides the following cleaning lemma that unifies Ref. \cite{williamson2024layer,lin2023geometrically,hastings2016weight,williamson2024low} and guarantees a relation between the code distances of the cone and its embedded code.
We elaborate on its application in the following sections -- specifically, Theorem \eqref{thm:layer-code-distance}, Theorem \eqref{thm:square-distance} and Proposition \eqref{prop:weight-cone-distance}.
It's worth mentioning that shortly after our manuscript, more cleaning lemmas have also been independently discovered by Refs. \cite{baspin2025fast,cowtan2025fast,zheng2025high}, albeit in a more restricted context that also falls under our framework.

\begin{lemma}[Cleaning]
    \label{lem:cleaning}
    Let $C$ be a height-2 cone with levels $C^{s},s=2,1,0$ and embedded complex $C^{\bg}$ and let $C^{s}$ be equipped with basis\footnote{See Definition \eqref{def:basis} and \eqref{def:cone}}. 
    Suppose that there exists $1\ge \alpha>0$ such that for any $\ell_2^{2}\in C_2^2$, there exists $\hat{\ell}_2^2\in C_2^2$ such that $\partial^2 \ell_2^{2} =\partial^2 \hat{\ell}_2^2$ and
    \begin{equation}
        \label{eq:isoperimetric}
        |\partial^2 \hat{\ell}_2^2|  \ge \alpha |g_2\hat{\ell}_2^2|
    \end{equation}
    Then
    \begin{equation}
        \label{eq:distance}
        \min_{\ell_1\in C_1: [\ell_1] \in H_1(\partial)\backslash 0} |\ell_1| \ge \alpha  \min_{\ell_1^1\in C_1^1:\llb \ell_1^1\rrb \in H_1(\partial^{\bg})\backslash 0} |\ell_1^1|
    \end{equation}
    In analogy to geometry \cite{tillich2000edge}, we refer to the condition in Eq. \eqref{eq:isoperimetric} as $(\partial^2,g_2)$ satisfying the \textbf{isoperimetric inequality} with \textbf{(isoperimetric) coefficient} $\alpha$.
    In particular, if $C^2$ is the direct sum of complexes equipped with basis\footnote{See Definition \eqref{def:direct-sum} for the convention of basis}, i.e., $C^2 = \bigoplus_{a} C^{2a}$ with $\partial^2 = \bigoplus_a \partial^{2a}$ and $(\partial^{2a},g_2)$ satisfies the isoperimetric inequality with coefficient $1\ge \alpha>0$ for all $a$, then Eq. \eqref{eq:distance} holds.
\end{lemma}
\begin{proof}
    See Appendix \eqref{proof:cleaning}
\end{proof}
\begin{remark}
    Regarding $C,C^{\bg}$ as CSS codes, we see that the left-hand-side (LHS) denotes the ($Z$-type) code distance $d_Z$ (see Example \eqref{ex:CSS}) of the cone $C$.
    The Cleaning Lemma then essentially states that the weight of $\ell_1$ can be \textit{cleaned} so that only its component in $C_1^1$ is relevant.
    Moreover, a lower bound of the code distance can be obtained by minimizing over representations $\ell_1^1\in C_1^1$ of nontrivial logical operators in the embedded code $C^{\bg}$.

    It's worth mentioning that the Cleaning Lemma is somewhat straightforward due to the isomorphism in Eq. \eqref{eq:height-2-cone-iso}, which states that every logical equivalence class has a representation which has 0 overlap with $C^2$, and thus the assumption in Eq. \eqref{eq:isoperimetric} aims to balance the weight of this representation with all possible representations.
\end{remark}

    


\subsection{Open Questions}
\begin{OQ}
Ref. \cite{portnoy2023local} provides a general construction of optimally embedding any LDPC code which admits a sparse $\dZ$ lift into Euclidean space.
Can the requirement of a sparse $\dZ$ lift be removed?
\end{OQ}

Roughly speaking, Ref. \cite{portnoy2023local} relies on the relatively advanced fact that if an LDPC code has a \textit{sparse $\dZ$ lift}, then the code can induce a triangulation on an 11D manifold \cite{freedman2021building}.
Further (barycentric) subdivisions of the 11D manifold can then be embedded into Euclidean space to induce a geometrically local quantum code.
However, since stabilizer codes are $\dF_2$ chain complexes, the requirement to first sparsely lift into $\dZ$, sudivide and then return back to $\dF_2$ chain complexes seems redundant in the application of quantum code embedding.
Indeed, in Section \eqref{sec:topology}, we show that our main result can be utilized to circumvent the necessity of an underlying manifold when subdividing a chain complex, and thus it raises the question whether the sparse $\dZ$ lift requirement can be removed.

\begin{OQ}
In this manuscript, we mostly focus on quantum embedding, in the sense that, given a CSS code $A$, we construct a mapping cone $C$ such that the embedded code is $A$.
However, can the converse be obtained, i.e., given an arbitrary CSS code $C$, does there exist a canonical embedded code $A$ within $C$?
\end{OQ}

\begin{OQ}
What are some general conditions in which the isoperimetric inequality \eqref{eq:isoperimetric} holds?
\end{OQ}

Despite the unified framework and corresponding Cleaning Lemma, proving the isoperimetric inequality is often a difficult task.
Section 5.4.1 of Ref. \cite{lin2023geometrically} builds upon Ref. \cite{tillich2000edge} and shows that the isoperimetric inequality is preserved under hypergraph products, which can be exploited if $\partial^2,g_2$ have well-behaved product structures of graphs.
However, the method only provides a lower bound on the isoperimetric coefficient, which may not necessarily be tight in many cases\footnote{For example, one may consider the hypergraph product of 1D chains of distinct lengths $L_1,L_2$, as given in Section 5.4.1 of Ref. \cite{lin2023geometrically}. If, say, $L_1 \gg L_2$, the lower bound on the coefficient will be $L_2/L_1\ll 1$.}.
Further insight into the isoperimetric inequality may be necessary for novels applications of quantum code embedding.

\medskip
\noindent
\textbf{Acknowledgments.}
We thank Jeongwan Haah for bringing to our attention Hastings' work on quantum weight reduction and thank Dominic J. Williamson for bringing up the topic of fault-tolerant logical measurement.
We thank Matthew B Hastings for general helpful discussions.
ACY was supported by the Laboratory for Physical Sciences at CMTC.
\section{Preliminaries}
\label{sec:prelim}
\subsection{Stabilizer and CSS codes}
\begin{definition}
A \textbf{stabilizer code} on finite $n$ qubits $\sH \cong(\dC^2)^{\otimes n}$ is defined as the subgroup $\sS$ of the Pauli group $\sP$, i.e., that generated by Pauli operators on $n$ qubits  and phase factors $\pm i,\pm 1$, such that $\sS$ is commuting and $-I\notin S$.
The subspace of \textbf{logical qubits} is that $V_\sS \subseteq \sH$ \textit{stabilized} by $\sS$, i.e., states $\psi \in \sH$ such that $s \psi = \psi$ for all $s \in \sS$.
The group $\sS^\perp$ of \textbf{logical operators} is the subgroup of $\sP$ which commutes with all operators in $\sS$, and the group of \textbf{equivalence classes of logical operators} by the quotient group $\sS^\perp/\bra i,\sS\ket$ where $\bra i,\sS\ket$ is the subgroup in $\sP$ generated by the phase factor $i$ and $\sS$.
We say that $\ell \in \sS^\perp$ is a \textbf{nontrivial} logical operator if its equivalence class is nonzero in $\sS^\perp/\bra i,\sS\ket$, i.e., $\ell \notin \bra i,\sS\ket$.
\end{definition}

While our primary interest lies in the error-protecting properties of the logical qubits, it is more convenient to examine the properties of the stabilizer code $\sS$ and its logical operators instead \cite{calderbank1997quantum,gottesman1996class,nielsen2010quantum}.
In fact, we further simplify the problem and study the abelianization of $\sP$.

\begin{definition}
\textbf{Abelianization} of the Pauli group $\sP$ is the homomorphism $\sA: \sP \to \dF_2^{2n}$ defined as follows: if $s \in \sP$ acts as $X,Y,Z$ on qubit $1\le i\le n$  (up to a phase), then its abelianization has elements $(1,0),(1,1),(0,1)$ at components $i,i+n$ in $\dF_2^{2n}$, respectively.
We call the image of the Pauli group $\sP$ under abelianization the \textbf{Pauli space} $P =\dF_2^{2n}$. 
To represent the commutation relations of elements in the Pauli group, we equip the Pauli space with the symplectic form 
\begin{equation}
    \lambda = \begin{pmatrix} & I \\I &\end{pmatrix}
\end{equation}
We say that $\ell,\ell'\in P$ are \textbf{commuting}, \textbf{anti-commuting} if $\ell \lambda \ell'=0,1$, respectively, and thus $\ell,\ell'\in \sP$ commute/anti-commute iff they commute/anti-commute after abelianization.
Moreover, given a subspace $S\subseteq P$, let $S^\perp$ denote the orthogonal complement of $S$ with respect to the symplectic form, i.e., all elements in $P$ commuting with $S$.
\end{definition}

\begin{remark}
\label{rem:abelian-correspondence}
Note that if $\sS$ is a stabilizer code, then its abelianization $S\equiv \sA \sS$ is a commuting subspace of the Pauli space such that $\sA:\sS\to S$ is an isomorphism.
Conversely, given a commuting subspace $S$ of the Pauli space $P$, there exists a commuting subgroup $\sS$ of $\sP$ with $-I\notin \sS$ such that $S$ is its abelianization.
Also note that $\sA \sS^\perp = (\sA \sS)^\perp = S^\perp$ for stabilizer codes $\sS$. 
In particular, $\sS^\perp/\bra i,\sS\ket \cong S^\perp/S$, and thus we will not differentiate between the two.
Moreover, since abelianization uniquely determines the Pauli operator up to a phase $\pm i,\pm 1$, we abuse terminology and also call $\ell\in S^\perp$ \textbf{logical operators} -- it is \textbf{nontrivial} if $\ell \notin S$.
\end{remark}

\begin{lemma}[Chap 10 of \cite{nielsen2010quantum}]
\label{lem:kd-correspondence}
Let $\sS$ be a stabilizer code on $n$ qubits with abelianization $S$ and $\dim S =n-k$. Define \textbf{code distance} $d$ as the minimum (Hamming) weight $|\ell|$ over nontrivial $\ell \in S^\perp\backslash S$.
Then $\dim V_\sS = 2^k$ and any error acting on $<d$ qubits can be corrected in the sense of Theorem 10.1 of Ref. \cite{nielsen2010quantum}.
\end{lemma}

The previous lemma implies that the essential properties of the code can be understood via analyzing the abelianizations.
The subclass of CSS codes is of particular interest, since they have natural representations as complexes.

\begin{definition}
\label{def:CSS}
Let $\sP_X,\sP_Z$ be the Pauli group generated by (pure) $X$- and $Z$-type Pauli operators on $n$ qubits, with abelianizations $P_X, P_Z$ so that $P_X\cong P_Z \cong \dF_2^{n}$.
A \textbf{Calderbank-Shor-Steane (CSS) code} is a stabilizer code $\sS$ which can be generated by $X$-type and $Z$-type Pauli operators, i.e., $\sS$ is generated by subgroups $\sS_X \equiv \sS\cap \sP_X$ and $\sS_Z \equiv \sS\cap \sP_Z$, or equivalently, the abelianization $S$ is the direct sum of its projection $S_X,S_Z$ on $P_X,P_Z$, respectively, i.e., 
\begin{equation}
    S = S_X\oplus S_Z
\end{equation}
Moreover, regard $S_X,S_Z$ be subspaces of $\dF_2^n$ and $S_X^\perp,S_Z^\perp$ as the orthogonal complements with respect to the standard inner product on $\dF_2^n$.
Define the collection of equivalence classes of \textbf{$X$-type logical operators} as $S_Z^\perp/S_X$ and that of \textbf{$Z$-type logical operators} as $S_X^\perp/S_Z$.
\end{definition}

\begin{lemma}
\label{lem:CSS-distance}
Let $\sS$ be a CSS code with abelianization $S=S_X\oplus S_Z$. Then
\begin{equation}
\frac{S^\perp}{S} \cong \frac{S_Z^\perp}{S_X} \oplus \frac{S_X^\perp}{S_Z}
\end{equation}
Moreover, if we define the $X$- and $Z$-type code distances as
\begin{equation}
    d_X = \min_{\ell \in S_Z^\perp\backslash S_X} |\ell|, \quad d_Z = \min_{\ell \in S_X^\perp\backslash S_Z} |\ell|
\end{equation}
Then the code distance $d$ of $\sS$ is $=\min(d_X,d_Z)$.
\end{lemma}

\begin{definition}[Basis]
    \label{def:parity-check}
    Consider a CSS code with abelianization $S=S_X\oplus S_Z$.
    Regard each $S_X,S_Z$ as subspaces of $\dF_2^n$ and choose an (ordered) basis of $S_X,S_Z$, respectively.
    The \textbf{support} of a basis elements are the indices with nonzero entry.
    Let $H_X,H_Z$ denote $n \times n_X$, $n\times n_Z$ matrices over $\dF_2$ such that each column is a basis element of $S_X,S_Z$, respectively.
    Then we say that the CSS code is generated by \textbf{parity check matrices} $H_X,H_Z$.
    Further denote
    \begin{itemize}
        \item The maximum weight of the columns in $H_X,H_Z$ is denoted $w_X \equiv |H_X|_{\mathrm{col}},w_Z \equiv |H_Z|_{\mathrm{col}}$, respectively, which depicts the max number of qubits any $X$-, $Z$-type generator acts on.  
        \item The maximum weight of the rows in $H_X,H_Z$ is denoted by $q_X \equiv |H_X|_{\mathrm{row}},q_Z \equiv |H_Z|_{\mathrm{row}}$ respectively, which depicts the max number of $X$-, $Z$-type generators that acts on any qubit.
    \end{itemize}
    We say that the code with given $H_X,H_Z$ is a \textbf{low-density parity-check code (LDPC)} if $w_X,w_Z,q_X,w_Z =O(1)$ in the large $n$ limit.
\end{definition}
\subsection{Chain complexes}
\begin{definition}
\label{def:chain-complex}
A \textbf{(chain) complex} $C$ is a family of $\dF_2$-vector spaces $C_{i}$, referred as the \textbf{degree} $i$ of $C$, together with $\dF_2$-linear maps $\partial_{i}:C_{i} \to C_{i-1}$, called the \textbf{differentials} of $C$, such that $\partial_{i}\partial_{i+1}:C_{i+1} \to C_{i-1}$ is zero. We write
\begin{equation}
    C = \cdots \to  C_{i} \xrightarrow{\partial_{i}} C_{i-1} \to \cdots
\end{equation}
Note that $0\subseteq \im \partial_{i+1} \subseteq \ker \partial_i \subseteq C_i$ for all $i$, and thus the \textbf{$i$-homology} of $C$ is defined as
\begin{equation}
    H_i(C)\equiv H_i (\partial) = \ker \partial_i/\im \partial_{i+1}
\end{equation}
Denote the equivalences class of $\ell_{i} \in C_{i}$ as $[\ell_{i}]\in H_i(\partial)$ so that the notation $[\cdots]$ is \textit{reserved} for the quotient map -- conversely, we may also write $[\ell_i] \in H_i(\partial)$ without specifying the representation $\ell_i$.
The complex $C$ is \textbf{exact} if $H_i(\partial) =0$ for all $i$ and two complexes are \textbf{quasi-isomorphic} if their homologies are isomorphic.
Adopt the following \textit{conventions}
\begin{itemize}
    \item The subscript in $\partial_{i}$ is often omitted, so that the condition is simplified as $\partial \partial =0$.
    \item Assume $C$ is of finite length, i.e., $C_i=0$ for all but finitely many $i$, and thus assume that $C$ ends at degree $i=0$, i.e., write $C_n \to \cdots \to C_0 $ (though $C_0$ is possibly 0) so that the \textbf{length} of $C$ is $n$.
\end{itemize}
\end{definition}

\begin{definition}[Basis]
\label{def:basis}
A complex $C$ is equipped with an \textbf{(ordered) basis} $\sB(C)$ if each degree $C_i$ is equipped with an ordered basis $\sB(C_i)$ --  we called the basis elements \textbf{$i$-cells}. 
The basis induces a fixed isomorphism $C_i \cong \dF_2^{n_i}$ and thus a well-defined \textbf{(Hamming) weight} $|\ell_{i}|$ for any $\ell_{i} \in C_{i}$.
Note that any $\ell_i\in C_i$ must be the $\dF_2$-span of basis elements and thus can be regarded as a subset $\ell_i\subseteq \sB(C_i)$ --  we call this the \textbf{natural identification}\footnote{
In particular, summation in $C_i$ corresponds to taking the symmetric difference and the weight of $\ell_i$ is the cardinality of $\ell_i$ as a subset}.
Adopt the \textit{convention} that $c_i\in \sB(C_i)$ and that $\sB(C_i)$ is the standard basis if $C_i = \dF_2^{n_i}$.
We say that two complexes with basis are \textbf{equal} if there exists a bijective mapping between the  $i$-cells for each $i$.
\end{definition}

\begin{remark}
    We emphasize the main results Theorem \eqref{thm:height-2} and \eqref{thm:height-n} do not require a basis.
\end{remark}

\begin{definition}[Cochain]
Let $C$ denote a complex with basis.
The isomorphism $C_i\cong \dF_2^{n_i}$ induces an inner product on $C_i$, and thus the transpose $\partial_i^T:C_{i-1} \to C_i$ is well-defined so that that \textbf{cochain complex}\footnote{The cochain complex can be defined in a more general setting \cite{hatcher2005algebraic,weibel1994introduction}, but for the purpose of CSS codes, the previous definition is sufficient. 
A similar simplification was also assumed in Ref. \cite{breuckmann2021balanced}.} is that given by
\begin{equation}
    C^T = \cdots \leftarrow C_{i} \xleftarrow{\partial_{i}^T} C_{i-1} \leftarrow \cdots
\end{equation}
The \textbf{$i$-cohomology} of $C$ is $H^i(\partial) = \ker \partial_{i+1}^T/\im \partial_i^T$ which is isomorphic to $H_i(\partial)$.
Note that for a length $n$ complex, $H^i(\partial) = H_{n-i}(\partial^T)$.


\end{definition}

\begin{remark}
    Note that the cochain complex is obtained by inverting all arrows. This also applies to Theorem \eqref{thm:height-2} (or Diagram \eqref{eq:height-2-diagram}).
\end{remark}

\begin{definition}[Adjacency]
Let $C$ be a complex with $i$-cells $c_i$.
We say that $c_{i},c_{i-1} $ are \textbf{adjacent}  if $\bra c_{i-1} |\partial c_{i}\ket \ne 0$, and write $c_{i} \sim^{C} c_{i-1}$ or $c_i \sim c_{i-1}$. 
Extend the definition and say that $c_{i},c_{j}$ with $i\ne j$ are \textbf{adjacent} if there exists $c_{m}$ for $m=i-1,...,j$ such that $c_{m+1}\sim c_{m}$ for all $m$, and write $c_{i} \sim c_j$.
The \textbf{($j$)-support} of $c_i$, denoted by $\supp_{j} c_i \equiv \supp_{j}^{C} c_{i}$ where $j\ne i$, is the collection of $c_j$ adjacent to $c_i$.
\end{definition}

\begin{example}[CSS Codes]
\label{ex:CSS}
Consider a CSS code with parity check matrices $H_X,H_Z$.
Define the following sequence
\begin{equation}
    C = \dF_2^{n_Z} \xrightarrow{H_Z} \dF_2^n \xrightarrow{H_X^T} \dF_2^{n_X} 
\end{equation}
Then $C$ is a complex with $\im H_Z = S_Z$  and $\ker H_X^T = S_X^\perp$ and $H_1(C),H^1(C)$ are the collection of equivalence classes of $Z$-, $X$-type logical operators, respectively.
Conversely, given a complex, i.e., $C=C_2 \to C_1 \to C_0$, equipped with basis, the parity matrices $H_Z,H_X^T$ can be defined via matrix representations of $\partial_2,\partial_1$, respectively, so that the code is given via Remark \eqref{rem:abelian-correspondence}.
In particular, the \textbf{$Z$-} and \textbf{$X$-type code distance} of $C$ can be defined as in Lemma \eqref{lem:CSS-distance}, and the \textbf{support} of $c_2$ or $c_0$ will always refer to the $1$-support in the context of CSS codes.
For convenience, denote the weight $w_Z(C),w_X(C),q_Z(C),q_X(C)$ via the following diagram
\begin{equation}
C_2 \xrightleftharpoons[q_Z]{w_Z} C_1 \xrightleftharpoons[w_X]{q_X} C_0
\end{equation}
Furthermore, given $c_2,c_0$, we shall refer to the \textbf{common qubits} -- denoted as $c_2\wedge c_0$ -- as intersection of supports of $c_2,c_0$, i.e., the collection of $c_1$ such that $c_2 \sim c_1 \sim c_0$.

\end{example}

\begin{example}(Classical Codes)
\label{ex:classical}
A classical code can be regarded as a CSS code with only one type, say $Z$-type, generators, and thus a classical code can be written as a complex of length $1$, i.e., $\dF_2^{n_Z} \to \dF_2^n$.
One of the most important classical codes is the \textbf{repetition code} on $L$ qubits, which is defined by the generators $Z_i Z_{i+1}$ where $i=1,...,L-1$.
We denote the corresponding complex as $R \equiv R(L)$ with differential $\partial^{R}$.
The 1-cells are denoted via half-integers $|i^+\ket$ for $i=1,...,L-1$ where $i^\pm =i\pm 1/2$, and 0-cells via integers $|i\ket$ for $i=1,...,L$ so that
\begin{equation}
    \partial^{R} |i^+\ket = |i\ket +|i+1\ket
\end{equation}
We implicitly assume that $|i\ket, |i^\pm\ket =0$ if the label is not within the previous parameters.
\end{example} 


\begin{definition}[Variations]
For our purposes, we will also need the following variations.
\begin{itemize}
    \item The \textbf{cyclic repetition code} $R^{\circlearrowleft} = R^{\circlearrowleft}(L)$ on $L$ qubits is the complex $\dF_2^{L} \to \dF_2^{L}$ with
    \begin{equation}
        \partial^{R^{\circlearrowleft}}|i^+\ket =|i\ket +|i+1\ket
    \end{equation}
    where the 1- and 0-cells are labeled via half-integers $|i^+\ket$ and integers $|i\ket$ where $i=1,...,L$, respectively, and addition of labels is computed modulo $L$.
    Note that $(R^{\circlearrowleft})^T =R^{\circlearrowleft}$.
    \item The \textbf{dangling repetition code} $R^{\multimap}=R^{\multimap}(L)$ on $L$ qubits is the complex $\dF_2^{L} \to \dF_2^{L}$ with 
    \begin{equation}
        \partial^{R^{\multimap}}|i^+\ket =|i\ket +|i+1\ket
    \end{equation}
    where the 1- and 0-cells are labeled via half-integers $|i^+\ket$ and integers $|i\ket$ where $i=1,...,L$, respectively, and implicitly assume that $|i\ket, |i^\pm\ket =0$ if the label is not within the previous parameters.
\end{itemize}
\end{definition}
\begin{lemma}[Repetition Code]
\label{lem:rep}
Let $R,R^{\circlearrowleft},R^{\multimap}$ denote the different variations of the repetition code. Then
\begin{equation}
\begin{aligned}[c]
    H_1(R) &=0  \\
    H_1(R^{\circlearrowleft}) &\cong \dF_2 \\
    H_1(R^{\multimap}) &\cong 0
\end{aligned}
\qquad
\begin{aligned}[c]
    H_0(R) &\cong \dF_2    \\
    H_0(R^{\circlearrowleft}) &\cong \dF_2 \\
    H_0(R^{\multimap}) &=0
\end{aligned}
\end{equation}
The unique basis element of $H_0(R)$ is given by $[|i\ket]$ for any $i=1,...,L$ and the unique basis element of $H^0(R)$ is $[\sR_0]$ where
\begin{equation}
    \sR_0 \equiv \sum_{i=1}^{L} |i\ket
\end{equation}
Note via the natural identification, $\sR_0$ can be regarded as all basis elements in $R_0$.
\end{lemma}

\subsection{Chain Complex Operations}
\begin{definition}[Height-($\le 2$) Cone]
\label{def:cone}
In this manuscript, we will be constructing sequences $C$ by \textbf{gluing} complexes $C^{s}$ together, in the sense of Theorem \eqref{thm:height-2} or Diagram \eqref{eq:height-2-diagram}.
Therefore, for simplicity, we write
\begin{equation}
\begin{tikzpicture}[>=implies,baseline]
\matrix(a)[matrix of math nodes, nodes in empty cells, nodes={minimum size=5mm},
row sep=2em, column sep=2em,
text height=1.5ex, text depth=0.25ex]
{C^{2} & C^{1} & C^{0}\\};
\draw[double,->,font=\scriptsize](a-1-1) -- node[above]{$g_{2}$} (a-1-2);
\draw[double,->,font=\scriptsize](a-1-2) -- node[above]{$g_{1}$} (a-1-3);
\draw[double,->,dashed,font=\scriptsize](a-1-1) to [out=-30,in=-150] node[below]{$p$} (a-1-3);
\end{tikzpicture}
\end{equation}
to denote $C$ obtained by setting $C_{i} = \oplus_s C^{s}_{i}$ and map $\partial:C_{i} \to C_{i-1}$ as the lower triangular matrix in Theorem \eqref{thm:height-2}. 
For simplicity, an arrow/term is omitted if the corresponding map/term is zero. 
We may pad levels with zeros, e.g., $0\to C^{s}$ or $C^{s} \to 0$, so that the all levels in the diagram are of the same length.
We write a similar definition for the construction of the height-1 cone $C=C^1 \Rightarrow C^0$.
We emphasize that it needs to be proven that the sequence $C$ is a complex, i.e., $\partial \partial =0$.
If $C^{s}$ are equipped with a basis $\sB(C^{s})$, then $C$ is equipped with the disjoint union $\sqcup_{s}\sB(C^{s})$ and we will define the gluing maps $g_1,g_2, p$ as the $\dF_2$ linear extension of specified actions on the basis.
If an action is not specified on some basis element $c^{s}_{i}$, then it maps the element  to $0$ \textit{by default}.

\begin{remark}
    Despite not necessary to understand this manuscript, it may be worth mentioning that the data for the height-2 cone is the first part of the data $\bra a, b, c\ket$ used for Toda brackets \cite{toda1962composition} where $a=g_2,b=g_1,c=0$.
    However, this connection fails in the generalization to height $n\ge 3$-cones.
\end{remark}

\begin{definition}[Tensor Product]
    \label{def:tensor-product}
    Let $C,D$ denote complexes with differentials $\partial^{C},\partial^{D}$. Then the \textbf{tensor product} is the complex $C\otimes D$ with degrees and differential
    \begin{equation}
        (C\otimes D)_{m} = \bigoplus_{i+j=m} C_{i} \otimes D_{j}, \quad \partial =\partial^{C} \otimes I \oplus I \otimes \partial^{D}
    \end{equation} 
    If $C,D$ have basis $\sB(C),\sB(D)$, then adopt the \textit{convention} that $C\otimes D$ has the basis at degree $m$ consisting of $c_i\otimes d_j$ over all $i+j=m$.
    The following lemma implies an analogous \textit{convention} can be adopted for the basis of $H_m(C\otimes D)$, provided that the homologies $H_i(C),H_i(D)$ are equipped with a basis for all $i$.
\end{definition}
\begin{lemma}[K\"{u}nneth Formula \cite{breuckmann2021balanced,weibel1994introduction}]
    \label{lem:Kunneth} 
    Let $C,D$ be complexes. Then 
    \begin{equation}
        H_m(C\otimes D) \cong  \bigoplus_{i+j=m} H_{i}(C) \otimes H_{j}(D)
    \end{equation}
\end{lemma}

\begin{example}[2D Toric Code \cite{tillich2013quantum,breuckmann2021balanced}]
    \label{ex:toric}
    Note that if $C,D$ are of length $n_C,n_D$, then its tensor product is of length $n_C+n_D$. In particular, the tensor product of two classical codes is a quantum code.
    For example, the 2D toric code on a torus, written as a complex $C$, is the tensor product of cyclic repetition codes, i.e., $C=R^{\circlearrowleft}\otimes R^{\circlearrowleft}$.
    The 2D toric code with alternating smooth and rough boundary conditions is the tensor product of classic repetition codes $R\otimes R^T$.
    The 2D toric code with smooth, rough boundary conditions is exactly the tensor product $R\otimes R, R^T\otimes R^T$, respectively.
    See Fig. \ref{fig:toric-boundaries}.
\end{example}
\begin{definition}[Direct Sum]
    \label{def:direct-sum}
    Let $C^{s}$ denote complexes with differentials $\partial^{s}$. Then the \textbf{direct sum} $C \equiv \oplus_{s} C^{s}$ is the complex consisting of degrees and differential
    \begin{equation}
        C_i = \bigoplus_{s} C^{s}_{i}, \quad \partial = \bigoplus_{s} \partial^{s}
    \end{equation}
    If $C^{s}$ is equipped with a basis $\sB(C^{s})$, then adopt the \textit{convention} that its direct sum has basis consisting of the disjoint union $\sqcup_{s} \sB(C^{s}_{i})$ at degree $i$.
    The following lemma also implies that an analogous \textit{convention} for the basis of $H_i(C)$, provided that the homologies of $C^{s}$ are equipped with a basis.
\end{definition}
\begin{lemma}
    \label{lem:direct-sum}
    Let $C^{s}$ be complexes with direct sum $C$. 
    Then 
    \begin{equation}
        H_i(C) = \bigoplus_{s} H_i(C^{s})
    \end{equation}
\end{lemma}

\begin{remark}
    \label{rem:direct-vs-prod}
    If $C =\oplus_{s} A$ denotes copies of a complex $A$ with basis $\sB(A)$, then we denote the basis elements in $\sB(C_i)$ via $|a_i;s\ket$ where $a_i\in \sB(A_i)$.
    If $A$ is the repetition code (or its variations), we write $|\sigma;s\ket$ instead of $||\sigma\ket;s\ket$ where $\sigma$ is an integer or half-integer.
    In particular, if $S$ denotes the $\dF_2$ vector space generated by a basis $|s\ket$ labeled via $s$, then it's clear that
    \begin{equation}
        \bigoplus_{s} A \cong A\otimes S 
    \end{equation}
    where we treat $S$ as a complex of length 0, and so we may also denote the basis elements as $a_i\otimes |s\ket$.
\end{remark}

\end{definition}
\section{Topological Codes}
\label{sec:topology}

This section will serve as a warm-up for upcoming sections, in which we show how the mapping cone can be applied to topological codes.
The exact definition of topological codes is somewhat vague, but it should at least include the following two classes.

The first class is the 2D toric code -- the paradigm for topological order -- with boundary conditions.
The existence of boundary conditions implies that it cannot be induced as the CW complex of 2D manifold and thus whether techniques developed from algebraic topology can be utilized is not apparent.
In particular, the rough and smooth boundary conditions of a 2D toric code \cite{bravyi1998quantum,kitaev2012models} are generally defined on finite square lattice.
However, similar constructions should also exist on other translationally invariant planar graphs, e.g., the honeycomb or triangular lattices.
This begs the question of how the boundary conditions on such lattices are structured and whether there is a natural isomorphism between distinct lattices with boundary conditions.

The second class are codes induced by discretizing a manifold. 
More specifically, a simplicial complex of a topological manifold induces a (finite) complex (over $\dF_2$) $C_{n} \to \cdots \to C_{0}$ with basis given by the simplices.
Choosing any subsequence of length 2, say $C_{i+1} \to C_{i} \to C_{i-1}$ for some $1<i<n$, thus defines a CSS code.
A manifold, however, can have multiple simplicial complexes, each inducing a distinct CSS code, and thus begs the question whether the resulting CSS codes are quasi-isomorphic in some natural manner.
In algebraic topology, this is answered by showing that the homology of any simplicial complex is isomorphic to the (uniquely defined) singular homology of the underlying manifold \cite{hatcher2005algebraic}.
However, one may wonder if it is possible to circumvent this construction since it relies on relatively in-depth knowledge regarding the point-set topology of the underlying manifold, and show directly that two simplicial complexes are isomorphic.

We comment that the following subsections can be read independently.

\subsection{2D Toric Code with Boundary Conditions}

\begin{figure}[ht]
\centering
\subfloat[\label{fig:toric-square}]{%
    \centering
    \includegraphics[width=0.47\columnwidth]{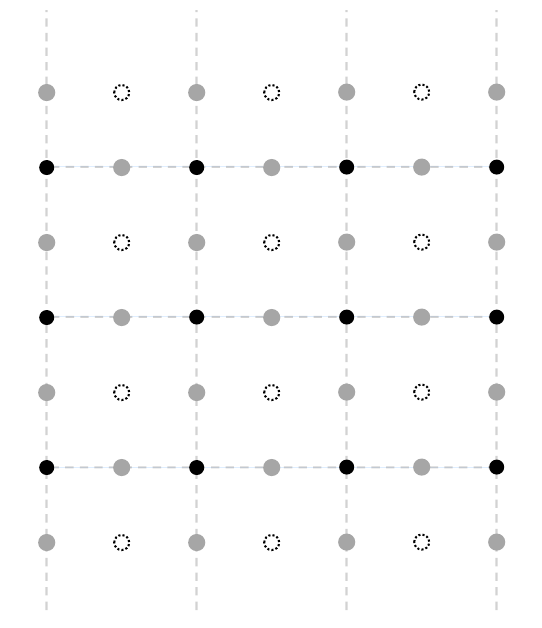}
}
\hspace{0.01\columnwidth}
\subfloat[\label{fig:toric-honey}]{%
    \centering
    \includegraphics[width=0.47\columnwidth]{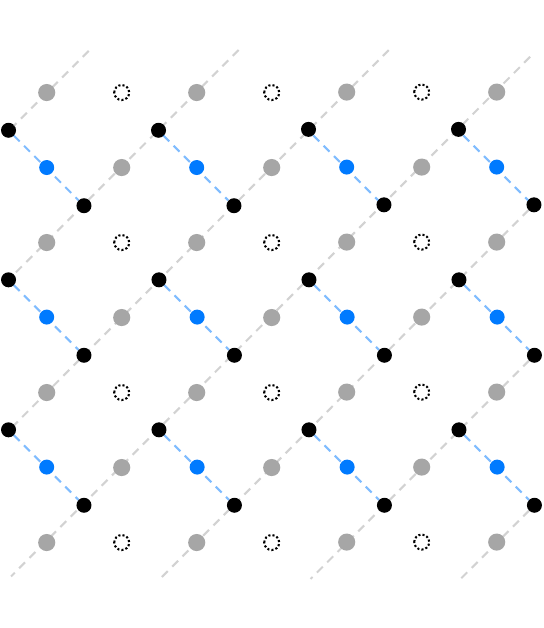}
}
\\
\subfloat[\label{fig:toric-triangle}]{%
    \centering
    \includegraphics[width=0.47\columnwidth]{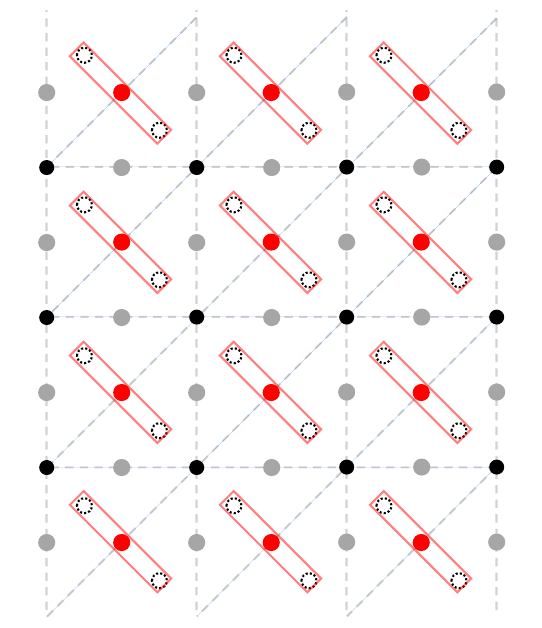}
}
\caption{Toric Code on Planar Graphs.
Edges of the lattice are denote by dashed lines.
Each edge hosts a qubit, while each vertex/plaquette hosts an $X$-/$Z$- type Pauli operator acting on adjacent qubits, denoted by a black/dashed dot, respectively.
(a) depicts the toric code (with alternating smooth and rough boundaries) on a finite square lattice where qubits are denoted by grey dots.
(b) depicts the toric code on the honeycomb lattice with boundary conditions corresponding to (a), where qubits are denoted by grey and blue dots.
(c) depicts the toric code on the triangular lattice with boundary conditions corresponding to (a), where qubits are denoted by grey and red dots.
Note that certain red dots and neighboring dashed dots are \textit{boxed} together to emphasize the construction in Eq. \eqref{eq:tri-lvl-2} in Theorem \eqref{thm:toric-triangular}.
}
\label{fig:toric-planar}
\end{figure}

By Example \eqref{ex:toric}, the toric code on a (finite) square lattice is described by the complex $\Tor^{\sq} = X \otimes Y$, where $X =R^{\circlearrowleft}(L_X), Y=R^{\circlearrowleft}(L_Y)$ if periodic boundary conditions are considered, and $X=R(L_X),Y=R(L_Y)^T$ if alternating smooth and rough boundary conditions are considered.
We adopt the notation that the 1-cells of $X$ are given by $|x_1\ket$ where $x_1$ is a half-integer and 0-cells are given by $|x_0\ket$ with $x_0$ is an integer -- the notation is similar for $Y$.

As discussed previously, the toric code $\Tor^{\sq}$ on the square lattice and that on a distinct lattice, say the honeycomb lattice $\Tor^{\hon}$, should be equivalent in some manner.
The equivalence mapping should be local in real space and independent of the boundary condition when applied to the \textit{bulk} of the lattice.
One way is to consider Fig. \ref{fig:toric-honey} and note that if the dashed blue line is \textit{compressed} to a single point, then the honeycomb lattice is (at least visually) reduced to a square lattice.
Using the main result in Theorem \eqref{thm:height-2}, this intuition can be formalized precisely as follows.
Let level 2, 1 be the complexes induced by $\Tor_2^\sq, \Tor_1^\sq$, respectively, and level 0 be obtained by \textit{grouping} elements on the blue dashed line together, i.e.,
\begin{align}
    \label{eq:hon-lvl-2}
    C^{2} &\equiv \Tor_2^\sq \to 0\to 0 \\
    C^{1} &\equiv 0 \to \Tor_1^\sq \to 0 \\
    C^{0} &\equiv 0 \to \Tor_0^\sq \otimes R(2)
\end{align}
Define the $g_1,p$ as follows
\begin{align}
        g_1 |x_1 y_0\ket &= |x_1^+ y_0 2\ket + |x_1^- y_0 1\ket \\
        g_1 |x_0 y_1\ket &= |x_0 y_1^+ 1\ket + |x_0 y_1^- 2\ket \\
        p |x_1 y_1\ket   &= (|x_1^+ y_1^+\ket + |x_1^- y_1^-\ket)\otimes |1^+\ket
        \label{eq:hon-homotopy}
\end{align}
Then we have the following statement
\begin{theorem}[Honeycomb]
    \label{thm:toric-honeycomb}
    Let $\Tor^{\sq} = X\otimes Y$ with $\partial^{\sq}$ be the toric code on the square lattice (with periodic or alternating boundary conditions).
    Let $C$ be the sequence obtained via the gluing procedure in Definition \eqref{def:cone}
    where $g_2=\partial_2^{\sq}$ and others terms are defined in Eq. \eqref{eq:hon-lvl-2}-\eqref{eq:hon-homotopy}.
    Then $C$ is a regular cone describing the toric code on the honeycomb lattice, i.e., $C=\Tor^{\hon}$, with embedded code $\Tor^{\sq}$.
\end{theorem}

\begin{proof}
    See Appendix \eqref{proof:toric-honeycomb}
\end{proof}

A similar mapping can be constructed for the triangular lattice, as sketched in Fig. \ref{fig:toric-triangle}, where elements in the red boxes are \textit{grouped} together.
The precise statement is then given as follows with definitions
\begin{align}
    \label{eq:tri-lvl-2}
    C^{2} &\equiv \Tor^{\sq}_{2} \otimes R(2)^T \to 0 \\
    C^{1} &\equiv 0 \to \Tor^{\sq}_{1} \to 0 \\
    C^{0} &\equiv 0 \to 0 \to \Tor^{\sq}_{0}
\end{align}
Define $g_2,p$ as follows
\begin{align}
        g_2 |x_1 y_1 1\ket &= |x_1^+ y_1\ket + |x_1 y_1^-\ket \\
        g_2 |x_1 y_1 2\ket &= |x_1^- y_1\ket + |x_1 y_1^+\ket \\
        p |x_1 y_1\ket\otimes |1^+\ket   &= |x_1^+ y_1^+\ket + |x_1^- y_1^-\ket
        \label{eq:tri-homotopy}
\end{align}
\begin{theorem}[Triangular]
    \label{thm:toric-triangular}
    Let $\Tor^{\sq} = X\otimes Y$ with $\partial^{\sq}$ be the toric code on the square lattice (with periodic or alternating boundary conditions).
    Let $C$ be the sequence obtained via the gluing procedure in Definition \eqref{def:cone}
    where $g_1 = \partial_1^{\sq}$ and other terms are defined in Eq. \eqref{eq:tri-lvl-2}-\eqref{eq:tri-homotopy}.
    Then $C$ is a regular cone describing the toric code on the triangular lattice, i.e., $C=\Tor^{\hon}$, with embedded code $\Tor^{\sq}$.
\end{theorem}
\begin{proof}
    See Appendix \eqref{proof:toric-triangular}
\end{proof}

\subsection{Barycentric Subdivision}

\begin{figure}[ht]
\centering
\subfloat[\label{fig:barycenter-in}]{%
    \centering
    \includegraphics[width=0.47\columnwidth]{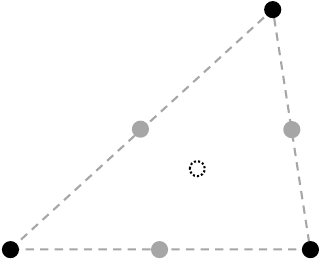}
}
\hspace{0.01\columnwidth}
\subfloat[\label{fig:barycenter-sub}]{%
    \centering
    \includegraphics[width=0.47\columnwidth]{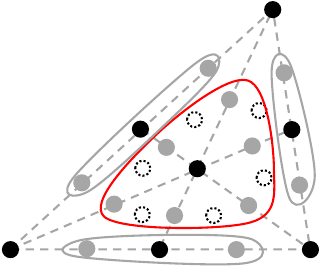}
}
\\
\subfloat[\label{fig:barycenter-general}]{%
    \centering
    \includegraphics[width=0.47\columnwidth]{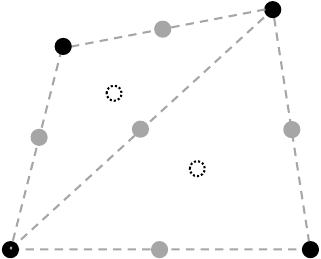}
}
\caption{Barycentric Subdivision of 2-simplex.
(a) depict the induced complex of a 2-simplex, whether the dashed, grey, black dots denote the 2-cell, 1-cells, 0-cells, respectively.
The dots also represent the barycenters of the plaquettes, edges and vertices, respectively.
(b) depicts the 2-cells, 1-cells, 0-cells of the barycentric subdivision of (a) via the dashed, grey, black dots, respectively.
The red and grey circle denote the levels $C^{2},C^{1}$, respectively.
(c) depicts a simplicial complex induced by two simplices.
}
\label{fig:barycenter}
\end{figure}

As discussed previously, any CSS code can be induced by a simplicial complex on a topological manifold.
Our goal is to establish an equivalence map between codes arising from different simplicial complexes on the same manifold, while avoiding the machinery of singular homology.
Rather than proving this in full generality, we focus on the case of a simplicial complex and its barycentric subdivision. 
This restriction does not result in a significant loss of generality, as the essential idea behind the original proof \cite{hatcher2005algebraic} involves iterated barycentric subdivisions, which renders the discrete structure increasingly close to a continuous one.


An $n$-simplex $\Delta$ with distinct vertices $|n\ket,....,|0\ket$ induces a complex $\Delta = \Delta_n \to \cdots \to \Delta_0$ with basis $\sB(\Delta_m)$ consisting of $|\bm{\ell}\ket \equiv |\ell_m,...,\ell_0\ket \equiv |\ell_m\ket \otimes \cdots \otimes |\ell_0\ket$ where $\ell_{m} > \cdots >\ell_{0}$, and differential
\begin{equation}
    \partial^{\Delta} |\bm{\ell}\ket = \sum_{i=0}^{m} |\hat{\bm{\ell}}_{{i}}\ket
\end{equation}
where $|\hat{\bm{\ell}}_i\ket\equiv |\ell_m,...,\ell_{i+1},\ell_{i-1},...,\ell_0\ket$ is the (decreasing) sequence with component $|\ell_{i}\ket$ removed.
As sketched in Fig. \ref{fig:barycenter-in}, the complex chain induced by a 2-simplex, for example, is such that the 0-cells are the vertices labeled as $|2\ket,|1\ket,|0\ket$, the 1-cells are the edges labeled as $|2,1\ket,|2,0\ket, |1,0\ket$ and the 2-cell is the plaquette labeled as $|2,1,0\ket$.
A general finite simplicial complex (see Fig. \ref{fig:barycenter-general}) is the \textbf{union}\footnote{In comparison, the direct sum can be regarded as the disjoint union.} complex induced by a collection of simplices, i.e., given collection of complexes $\Delta^s$ induced by simplices, their union is that $\Delta$ with degree $\Delta_m = \sum_{s} \Delta^{s}_{m}$ and basis $\sB(\Delta_m)=\cup_{s} \sB(\Delta^{s}_{m})$ and differential $\partial = \sum_{s} \partial^{s}$, and we refer to the $i$-cells of a simplicial complex as the \textbf{$i$-simplices}.
In particular, the 0-simplices is the (not necessarily disjoint) union of the vertices of the collection of simplices.

The \textbf{barycentric subdivision} of a simplicial complex $\Delta$ is the simplicial complex $C= C_{n} \to \cdots \to C_{0}$ with basis $\sB(C_m)$ consisting of $\| \bm{\ell}\ket = \|\ell_{m},...,\ell_0\ket$ where $\|\ell_{i}\ket$ is an $\ell_{i}$-simplex of $\Delta$, $\ell_{m}>\cdots >\ell_{0}$ and $\| \ell_m \ket \sim^{\Delta} \cdots \sim^{\Delta} \|\ell_0\ket$, and differential
\begin{equation}
    \partial \| \bm{\ell}\ket = \sum_{i=0}^{m} \|\hat{\bm{\ell}}_{i}\ket
\end{equation}
As sketched in Fig. \ref{fig:barycenter}, one may label each $i$-simplex in $\Delta$ by their barycenter so that the $i$-simplices of the subdivision are naturally labeled by the subsets of the barycenters.
Then $C$ with differential $\partial$ is naturally an height-$n$ cone, i.e., the differential is lower-triangular with respect to levels $C^{s}:C^{s}_{s} \to \cdots \to C^{s}_{0}$ defined as that with basis $\sB(C^{s}_{m})$ consisting of $\| \bm{\ell}\ket \in \sB(C_m)$ where $\ell_{m} =s$ and differential
\begin{equation}
    \partial^{s} \| \bm{\ell}\ket = \sum_{i=0}^{m-1} \|\hat{\bm{\ell}}_{i}\ket
\end{equation}
In particular, the levels $C^{s}$ are natural in the sense of Fig. \ref{fig:barycenter-sub}, or more specifically, we \textit{claim} that (where the proof is postponed to that of Theorem \eqref{thm:barycentric-subdivision}) $H_{r}(C^{s}) =0$ for all $r <s$ and $H_{s}(C^{s}) \cong C_{s}$ with basis
\begin{equation}
    [\| s\ket] \equiv \sum_{\|m\ket,m<s: \| s\ket \sim^{\Delta} \| s-1\ket \sim^{\Delta} \cdots \sim^{\Delta} \| 0\ket} \|s,s-1,...,0\ket,
\end{equation}
where we adopt the convention that $\|m\ket$ are $m$-simplicies of $\Delta$ for all $m$.
Similar to Diagram \eqref{eq:height-n-diagram}, the diagram for the barycentric subdivison $C$ can thus be written as 

\begin{equation}
\label{eq:barycenter-diagram}
\begin{tikzpicture}[baseline]
\matrix(a)[matrix of math nodes, nodes in empty cells, nodes={minimum size=5mm},
row sep=1.5em, column sep=1.5em,
text height=1.5ex, text depth=0.25ex]
{ C^{n}_{n}     &\cdot & \cdot  & C^{n}_{0}   \\
  C^{n-1}_{n-1} &\cdot & C^{n-1}_{0} &\\
  \cdot         &\cdot &  &  \\
  C^{0}_{0}     & & &\\
};
\path[->,font=\scriptsize]
(a-1-1) edge (a-1-2)
(a-1-2) edge (a-1-3)
(a-1-3) edge (a-1-4)
(a-2-1) edge (a-2-2)
(a-2-2) edge (a-2-3)
(a-3-1) edge (a-3-2)
(a-1-1) edge (a-2-1)
(a-1-2) edge (a-2-2)
(a-1-3) edge (a-2-3)
(a-2-1) edge (a-3-1)
(a-2-2) edge (a-3-2)
(a-3-1) edge (a-4-1);
\end{tikzpicture}
\end{equation}
where only diagonal $\partial^{s}$ and subdiagonal matrix elements of $\partial$ are shown as arrows.
Moreover, by the previous claim, $C$ is a regular cone.
It's then straightforward to check that $\Delta$ is the embedded complex within $C$ so that by Theorem $\eqref{thm:height-n}$, the complexes are quasi-isomorphic, i.e., $H_m(C) \cong H_m(\Delta)$ for all $m$.
This discussion is collected as follows.
\begin{theorem}
    \label{thm:barycentric-subdivision}
    Let $\Delta$ be a (finite) simplicial complex with barycentric subdivision $C$.
    Then $C$ is a regular height-$n$ cone with levels $C^{s},s=0,...,n$.
    In particular, $\Delta$ and $C$ are quasi-isomorphic.
\end{theorem}

\begin{proof}
    See Appendix \eqref{proof:barycentric-subdivision}
\end{proof}
\section{Embedding into Euclidean Space}
\label{sec:euclidean}

In this section, we show how previous constructions of embedding LDPC codes into Euclidean space  naturally fit into our framework.
We comment each subsection can be read independently of one another.

\subsection{Layer Code \cite{williamson2024layer}}
\label{sec:layer-code}

\begin{figure}[ht]
\centering
\subfloat[\label{fig:toric-alter}]{%
    \centering
    \includegraphics[width=0.45\columnwidth]{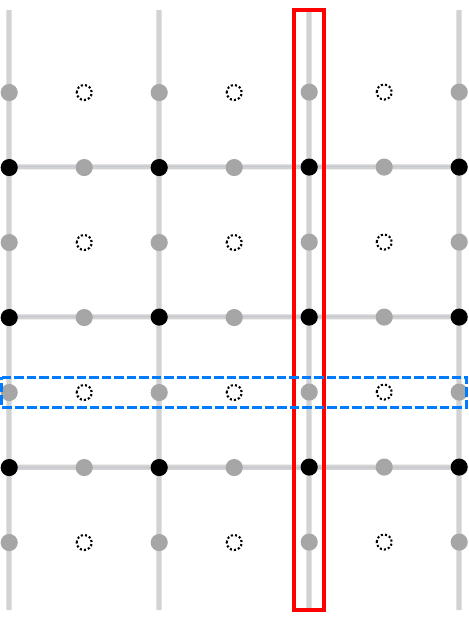}
}
\hspace{0.03\columnwidth}
\subfloat[\label{fig:toric-rough}]{%
    \centering
    \includegraphics[width=0.45\columnwidth]{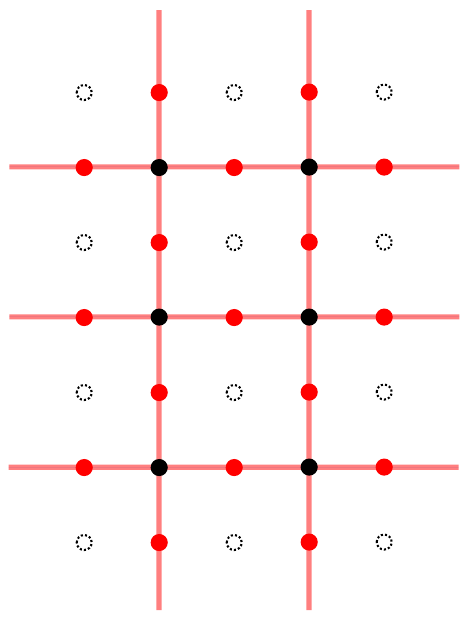}
}
\\
\subfloat[\label{fig:toric-smooth}]{%
    \centering
    \includegraphics[width=0.45\columnwidth]{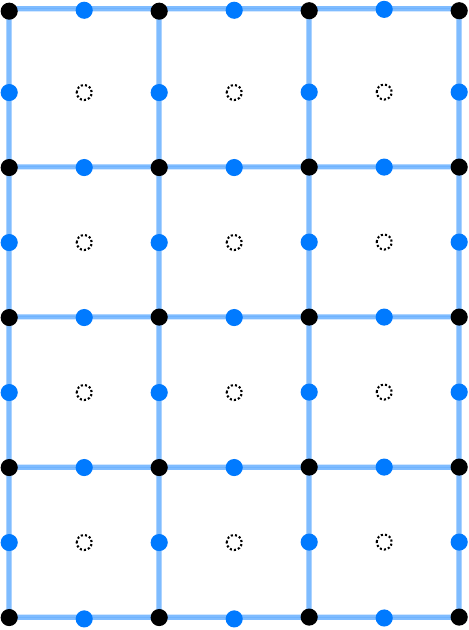}
}
\caption{Toric Codes with Boundaries. 
Each edge hosts a qubit, while each vertex (plaquette) hosts an $X$-, ($Z$-) type Pauli operator acting on adjacent qubits, denoted by a black (dashed) dot, respectively.
(a) depicts the toric code with alternating smooth and rough boundaries on a finite square lattice, and the qubits are denoted by grey dots.
$Z$ ($X$) Pauli operators acting on the qubits contained in the red (blue) rectangles denotes an example of nontrivial $Z$- ($X$-) type logical operator, respectively.
(b), (c) depict the toric code with rough, smooth boundaries on a finite square lattice, where qubits are denoted by red, blue dots, respectively.
(b), (c) do not host nontrivial logical operators.
}
\label{fig:toric-boundaries}
\end{figure}

\begin{figure}[ht]
\centering
\subfloat[\label{fig:layer-gluing-2}]{%
    \centering
    \includegraphics[width=0.95\columnwidth]{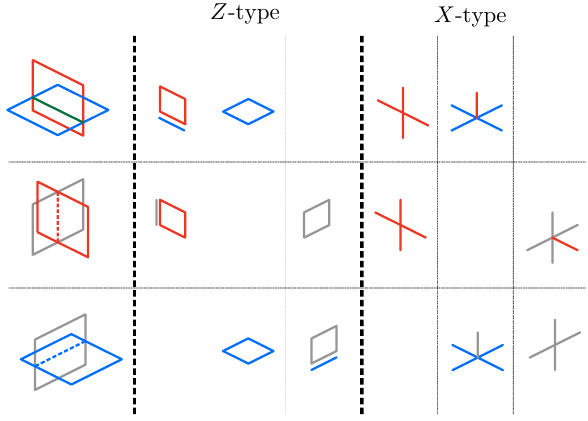}
}
\\
\subfloat[\label{fig:layer-gluing-3}]{%
    \centering
    \includegraphics[width=0.95\columnwidth]{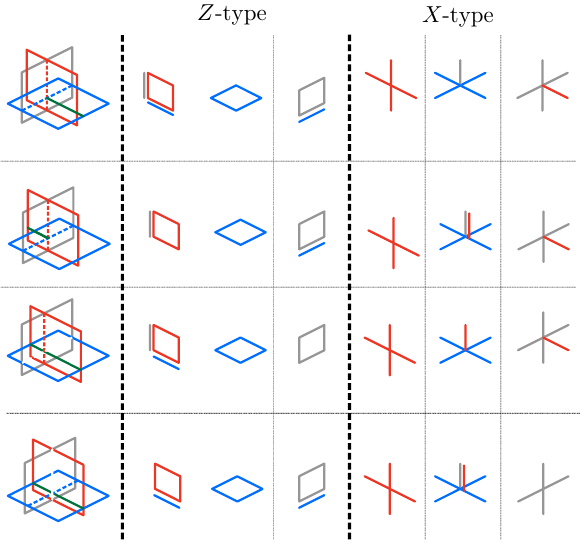}
}
\caption{Gluing of Toric Codes. 
(a) tabulates the modification of toric codes parity checks when two planes intersect, i.e., a $Z$- (red) or $X$- (blue) type plane interacting with a qubit plane (grey), or a $Z$- (red) and $X$- (blue) type plane interacting via a \textit{string defect} (green). 
For example, in row two, the $Z$-type plaquette operator of the $Z$-type plane (red) along the intersection (dashed red) is modified to also act on the corresponding edge (grey) of the qubit plane, while the $Z$-type plaquette operators of the qubit plane (grey) is not modified. (b) tabulates the modification of toric code parity checks when three planes intersect.
}
\label{fig:layer-gluing}
\end{figure}
\subsubsection{Quick Review}
Let us begin by reviewing the construction in Ref. \cite{williamson2024layer}, in which the authors provide a construction of embedding any CSS code into $\dR^3$.
Given a CSS code with parity check matrices $H_Z, H_X$, Ref. \cite{williamson2024layer} first (see Fig. \eqref{fig:layer-logical-example})
\begin{enumerate}[label=\arabic*)]
    \item Replaces each qubit with a toric code with alternating rough/smooth boundary conditions, i.e., Fig. \ref{fig:toric-alter}, in the $xz$ plane of $\dR^3$ (referred to as the \textbf{qubit plane}), so that the $Z$- ($X$-) type logical operators are along the $z$ ($x$) direction, respectively.
    In particular, there are $n$ parallel qubit planes ordered along the $y$ direction, so that qubit $i$ corresponds to plane $i$ in the positive $y$ direction. 
    \item Replaces each $Z$-type generator with a toric code with rough boundary conditions, i.e., Fig. \ref{fig:toric-rough}, in the $yz$ plane (referred to as the \textbf{$Z$-type plane}) so that there are $n_Z$ parallel copies.
    \item Replaces each $X$-type generator with a toric code with smooth boundary conditions, i.e., Fig. \ref{fig:toric-smooth}, in the $xy$ plane (referred to as the \textbf{$X$-type plane}), so that there are $n_X$ parallel copies.
\end{enumerate}

Up to this point, the toric codes should be regarded as ``floating" in $\dR^3$, and thus to include interactions, Ref. \cite{williamson2024layer} performs the following \textit{gluing} procedure.
If a $Z$-type generator acts on qubit, which does not interact with any $X$-type generator, then the corresponding planes are \textit{glued} along the $z$ direction, in the sense that the $Z$-, $X$-type operators of the toric codes are modified along the intersection, as shown in row two of Fig. \ref{fig:layer-gluing-2}.
The case is similar for when an $X$-type generator acts on qubit, which does not interact with any $Z$-type generator, and shown in row three of Fig. \ref{fig:layer-gluing-2}.

The complicated part (which we will show corresponds to $p$ in Theorem \eqref{thm:height-2}) originates from the fact that qubits can interact with both $Z$- and $X$-type generators.
In this case, Ref. \cite{williamson2024layer} introduces \textit{string defects} so that the result after embedding remains a stabilizer code.
More specifically, given a $Z$-type $\ell_Z$ and $X$-type generator $\ell_X$, the collection of qubits which interacts with both $\ell_Z,\ell_X$ must have even cardinality (since they commute) and thus their common set of qubits can be grouped into pairs with respect to the ordering along the $y$ direction, e.g., if qubits $1,6,7,9$ consist of those which interact with both $\ell_Z,\ell_X$, then they are paired as $(1,6),(7,9)$.
For each pair $(i,j)$ of qubits, the $Z$- and $X$-type planes corresponding to $\ell_Z,\ell_X$ must interact with qubit planes $i,j$ in the previous gluing sense, and thus a \textit{string defect} is introduced between qubit planes $i,j$ along the $y$ direction where the $Z$- and $X$-type planes intersect.
Specifically, along the $y$ axis and between qubit planes $i,j$, the toric code parity checks of the corresponding  $Z$- and $X$-type plane (and possibly qubit planes $i,j$) are modified near the string defect, as indicated by the green line in Fig. \ref{fig:layer-gluing}.
We note that row one and two of Fig. \ref{fig:layer-gluing-3} depict the start $i$ and end $j$ qubit plane, respectively, while row three of Fig. \ref{fig:layer-gluing-2} corresponds to the case where the $Z$-type plane interacts with some qubit plane $m$ between $i,j$, while the $X$-type plane ``passes" through (does not act on qubit $m$), and vice verse for row four.

It's worth mentioning that the construction is somewhat vague regarding certain details, such as the exact sizes of the planes, the exact parity checks involved in the gluing procedure at boundaries, etc.
We do not dwell on these issues since, shortly, we will define the construction from an algebraic perspective using Definition \eqref{def:cone}.
In fact, the brief overlook of tedious details allows us to understand the logical operators in a more intuitive manner.
Recall that due to the toric code parity checks, the $Z$-type logical operators must locally be paths on the square lattice, while the $X$-type logical operators are locally paths on the dual lattice \cite{kitaev2003fault}. 
Moreover, the paths corresponding to $Z$-, $X$-type logical operators must possess ends only on the rough, smooth boundaries, respectively \cite{bravyi1998quantum}.
Since the embedding procedure only modifies the parity checks near intersections, the same must also hold for the constructed code everywhere except along the intersection.
The behavior of logical operators near intersections is then tabulated in Fig. \ref{fig:layer-logical-rules}, and an example of a $Z$-type logical operator is shown in Fig. \ref{fig:layer-logical-example}.

\begin{figure}[ht]
\centering
\subfloat[\label{fig:layer-logical-rules}]{%
    \centering
    \includegraphics[width=0.55\columnwidth]{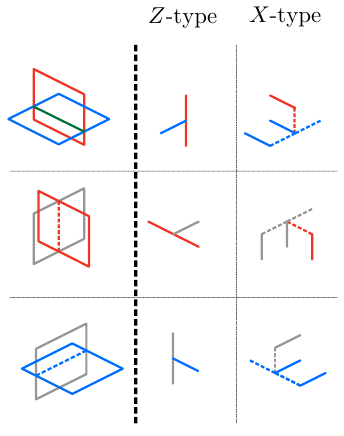}
}
\\
\subfloat[\label{fig:layer-logical-example}]{%
    \centering
    \includegraphics[width=0.8\columnwidth]{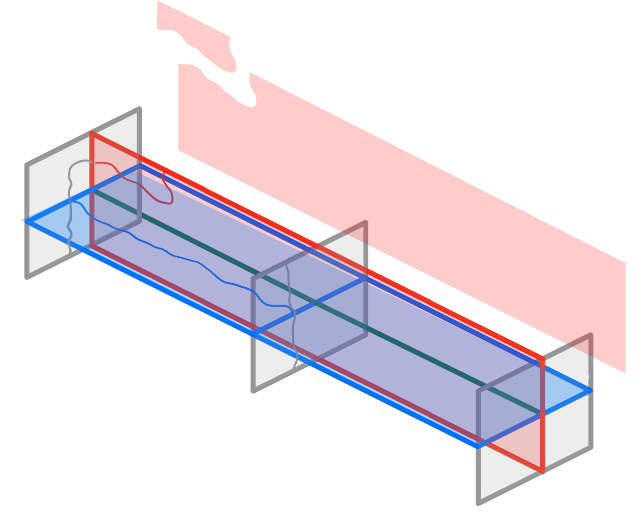}
}
\caption{Logical operators near Intersections. 
(a) tabulates the behavior of $Z$- and $X$-type logical operators which visit the intersections of toric planes. The dashed lines are to indicate that the $X$-type logical operators are paths on the dual lattice.
The first row, for example, indicates that a $Z$-type logical operator in the $Z$-type plane (red) passing through the string defect (green) must also exit into the $X$-type plane (blue).
(b) shows an example of the embedding procedure for the CSS code with generators $XXX$ and $ZIZ$. The squiggly lines denote a possible $Z$-type logical operator $\ell_1^1$ (induced by logical $ZZI$) based on the rules in Fig. \ref{fig:layer-logical-rules}, where the colors indicate which plane the path live in.
The two detached red $yz$-planes denote possible $\ell_2^2$ to use to \textit{clean} the logical operator $\ell_1$ so that its projection in $C^2_1$ is zero.
}
\label{fig:layer-logical}
\end{figure}
\subsubsection{Algebraic Construction}
We now formulate the construction within the framework of Theorem \eqref{thm:height-2}.
For the remainder of this subsection, let $A=A_2 \to A_1 \to A_0$ denote a complex of given CSS code with basis and label the $2,1,0$-cells via $|x_0\ket,|y_0\ket,|z_0\ket$ where $x_0 =1,...,n_Z$ and $y_0 =1,...,n$ and $z_0 =1,...,n_X$, respectively\footnote{The unfortunate notational mismatch between $|x_0\ket$ and $n_Z$ is to remain consistent with Ref. \cite{williamson2024layer}.}.

As in Example \eqref{ex:toric}, the toric code is the tensor product of classical repetition codes.
In particular, let $X =R(n_Z),Y=R(n),Z =R(n_X)$ and label the basis\footnote{Note that we have slightly abused notation and used $|x_0\ket$ for both $X_0$ and $A_2$.} of $X_i$ by $|x_i\ket$ for $i=1,0$, and similarly for $Y,Z$.
Then define the $Z$-type, qubit, $X$-type planes as
\begin{align}
    \label{eq:layer-code-lvls}
    C^{2} &= \bigoplus_{x_0 = 1}^{n_Z}  Y^T \otimes Z^T \\
    C^{1} &= \bigoplus_{y_0 = 1}^{n}  X \otimes Z^T \\
    C^{0} &= \bigoplus_{z_0 = 1}^{n_X} X \otimes Y
\end{align}
so that $C^{s}$ are complexes with differential $\partial^{s}$.
In particular, $C^{2}$ depicts $n_Z$ $Z$-type plane, one for each $Z$-type generator of the original code $A$, and thus $|\cdots;x_0\ket$ denote the basis element in $Z$-type plane $x_0=1,...,n_Z$.
The case is similar for $C^{1},C^{0}$.
Define the gluing maps $g_2,g_1$ via 
\begin{align}
    g_2 |y_0 z_i; x_0\ket &= 1\{|y_0\ket \sim^{A} |x_0\ket\}|x_0 z_i; y_0\ket \\
    g_1 |x_i z_0; y_0\ket &= 1\{|y_0\ket \sim^{A} |z_0\ket\}|x_i y_0; z_0\ket
\end{align}
Define the \textit{string defect} as follows. For repetition code $R$, let $R[i,j)$ induced by integers $i<j$ be the collection\footnote{If $R$ is regarded as a graph with edges and vertices given by the 1-cells and 0-cells, respectively, then $R[i,j)$ induces a string-like path on the graph between $i,j$.} of integers and half-integers $s$ such that $i\le s <j$.  
Moreover, if $S \subseteq \{1,...,L\}$ is an even subset with ordered elements $i_{1} <\cdots <i_{2m}$, we denote the $R[S)$ as the union of $R[i_{2s-1},i_{2s})$ over $s=1,...,m$.
Define the map $p$ via
\begin{equation}
\label{eq:layer-code-chain-homotopy}
    p|y_i z_0;x_0\ket =1\{|y_i\ket\in Y[|x_0\ket \wedge |z_0\ket)\}|x_0 y_i^+ ;z_0\ket 
\end{equation}
where $|x_0\ket \wedge |z_0\ket$ is the collection of common qubits $|y_0\ket\in \sB(A_1)$.

\begin{theorem}[Layered Code]
\label{thm:layer-code}
Given CSS code with complex $A$ and basis $\sB(A)$, define chain complexes $C^{s},s=2,1,0$, gluing maps $g_2,g_1$ and $p$ as in Eq. \eqref{eq:layer-code-lvls}-\eqref{eq:layer-code-chain-homotopy}, and define the cone $C$ via Definition \eqref{def:cone}.
Then $C$ is a regular cone with embedded code $A$.
\end{theorem}

\begin{proof}
See Appendix \eqref{proof:layer-code}
\end{proof}
\subsubsection{Code Distance}
So far, we have shown that the embedding procedure preserves logical qubits (operators), and thus the next goal is to show that the mapping cone preserves code distance in some manner.

\begin{theorem}[Layered Code Distance]
    \label{thm:layer-code-distance}
    Assume the same hypothesis as in Theorem \eqref{thm:layer-code} so that $C^{\bg}= A$ is the embedded code within $C$.
    Let $d_Z^{A},d_Z^{1}$ be the $Z$-type code distance of $A,C^{1}$, respectively, and similarly define the $X$-type code distance.
    Let $w_Z,w_X \ge 2$ denote the weights of $A$ (as defined in Example \eqref{ex:CSS}), respectively.
    Then the $Z$- and $X$-type code distance $d_Z,d_X$ of the 3-level cone $C$ is bounded below by
    \begin{equation}
        d_\alpha \ge \frac{2}{w_\alpha} d^{A}_\alpha d^{1}_\alpha, \quad \alpha =Z,X
    \end{equation}
    Note that $d^{1}_Z=n_X$ and $d^{1}_X=n_Z$ for the 2D toric code $C^{1}$ \cite{tillich2013quantum}.
\end{theorem}

If $(\partial^2,g_2)$ satisfies the isoperimetric inequality \eqref{eq:isoperimetric} with coefficient $2/w_Z$, then the statement follows directly from the Cleaning Lemma \eqref{lem:cleaning} -- see Fig. \ref{fig:layer-logical-example} for illustration.
However, due to the 2D structure of the $Z$-type planes $C^{2}$, proving the isoperimetric inequality directly is somewhat difficult\footnote{The isoperimetric inequality can be proven using the product structure of $Y^T\otimes Z^T$ and Section 5.4.1 of Ref. \cite{lin2023geometrically}, albeit the coefficient can only be determined up to $O(1/w_Z) < 2/w_Z$.}, and thus we provide an alternative route, though the underlying intuition is guided by the Cleaning Lemma.



\begin{proof}
See Appendix \eqref{proof:layer-code-distance}

\end{proof}

\subsection{Square Complexes \cite{lin2023geometrically}}
\label{sec:square}

\subsubsection{Quick Review}

The authors of Ref. \cite{lin2023geometrically} provide a construction of embedding any good LDPC code obtained from balanced products \cite{breuckmann2021balanced} into $\dR^{D}$ for any $D\ge 3$.
The embedding is optimal in the sense that the output code saturates the BPT bounds \cite{bravyi2009no,bravyi2010tradeoffs}, possibly up to poly-logarithmic corrections. 
The poly-log corrections originate from Ref. \cite{portnoy2023local}, though Ref. \cite{lin2023geometrically} claims that the corrections can be removed and a proof will be provided in an upcoming paper.

Roughly speaking, the embedding procedure first associates a hypergraph product \cite{tillich2013quantum,breuckmann2021balanced} to the balanced product code.
The faces of a hypergraph product are \textit{squares} and thus can be subdivided indefinitely into $L\times L$ sub-squares.
A CSS code is then extracted from the subdivision of the hypergraph product.
The subdivision is then mapped into Euclidean space via a local map in the following sense.
Given a CSS code with associated complex $C$ with basis $\sB(C)$, the map $f: \sB(C_{1})\to \dR^{D}$ is \textbf{local} if the following holds in the large $n=|\sB(C_1)|$ limit;
\begin{enumerate}[label=\arabic*)]
    \item Qubit density is finite, i.e., $|f^{-1}(x)| =O(1)$ for all $x\in \dR^{D}$
    \item Parity checks are local in $\dR^{D}$, i.e., if $c_1,c_1'\in \sB(C_1)$ adjacent to some $c_2$ (or some $c_0$), then $\norm{f(c_1) -f(c_1')}_2 =O(1)$
\end{enumerate}

Since the map $f$ can be regarded as a labeling of the basis elements $\sB(C_{1})$ using points\footnote{Or more accurately, points in $\dR^D \times \Lambda$ where $\Lambda$ is a set of cardinality $O(1)$.} in $\dR^D$, it does not affect the quantum dimension $k$ or code distance $d$. Hence, we will not be concerned with its construction and refer to Ref. \cite{portnoy2023local,lin2023geometrically} for the interested reader.
Instead, we will elaborate on the subdivision construction, whose procedure can be made explicit using the framework of Theorem \eqref{thm:height-2} and holds for all \textit{square complexes} as defined below.
\begin{definition}
    \label{def:square}
    Let $C=C_{2}\to C_{1} \to C_{0}$ be a complex equipped with $i$-cells $c_i$.
    Then $C$ is \textbf{square complex} if for every adjacent basis elements $c_2\sim c_0$, there exists exactly two distinct 1-cells $c_1^{\hor},c_1^{\ver}$ such that $c_2 \sim c_1^{\hor},c_1^{\ver} \sim c_0$.
    We call $c_1^{\hor},c_1^{\ver}$ \textbf{horizontal} and \textbf{vertical} with respect to $c_2,c_0$.
    We define a \textbf{square} to be a pair $c_2,c_0$.
\end{definition}

\subsubsection{Algebraic Construction}

For the remainder of this section, let $A=A_2\to A_1 \to A_0$ be a given square complex with $i$-cells $a_i$.
The toric code, for instance, is a square complex, and thus shall be an example depicted in Fig. \ref{fig:square}.
Let $R^{\multimap} = R^{\multimap}(L)$ be the dangling repetition code in Example \eqref{ex:classical} where we seek to construct the $L$-subdivision of $A$.
Similar to Ref. \cite{williamson2024layer} elaborated in Subsection \eqref{sec:layer-code}, the $L$-subdivision will be a height-2 cone $C$ obtained via gluing levels $C^{2},C^{1},C^{0}$ together.
However, what's different is that each level itself $C^{2},C^{1}$ will be obtained as the cones of gluing further levels together -- the subdivision $C$ is a cone of cones.
The real space procedure is sketched in Fig. \ref{fig:square-glue}, while its homological counterpart is shown in Diagram \eqref{eq:square-diagram}.

\begin{figure}[ht]
\centering
\subfloat[\label{fig:square-toric}]{%
    \centering
    \includegraphics[width=0.39\columnwidth]{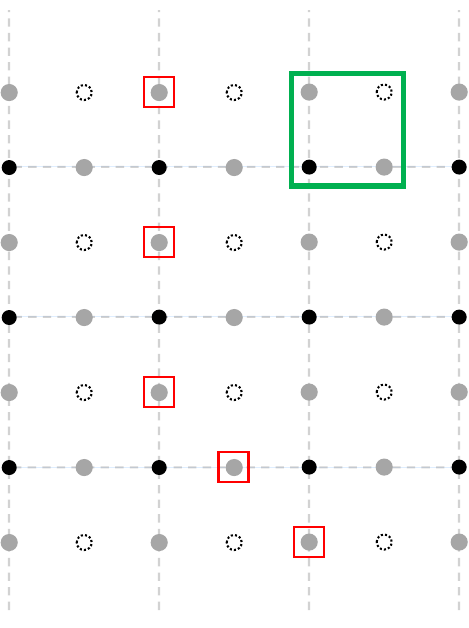}
}
\hspace{0.01\columnwidth}
\subfloat[\label{fig:square-glue}]{%
    \centering
    \includegraphics[width=0.5\columnwidth]{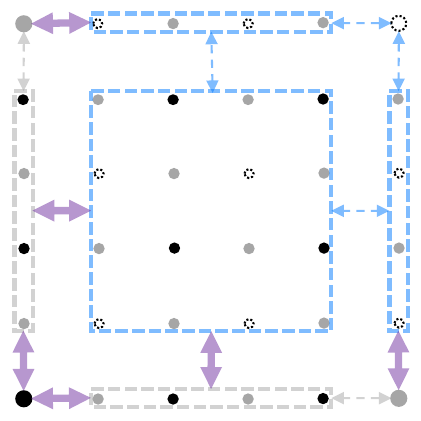}
}
\\
\subfloat[\label{fig:square-subdivision}]{%
    \centering
    \includegraphics[width=0.39\columnwidth]{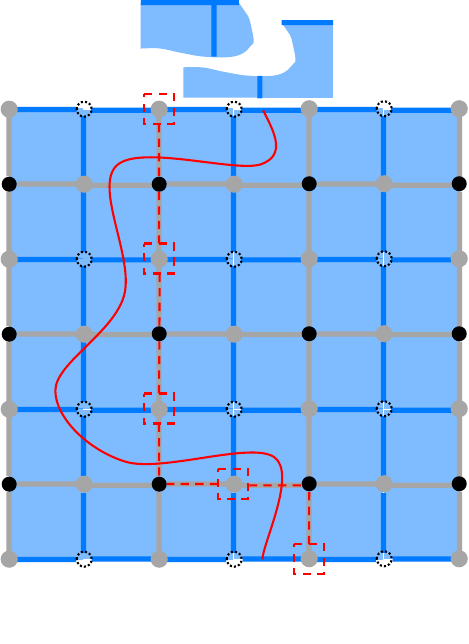}
}
\caption{Subdivision. 
(a) depicts the toric code with alternating smooth and rough boundaries on a finite square lattice.
The dashed guidelines denote edges hosting qubits, which will be denoted by grey dots instead.
The $X$-, $Z$-type Pauli operators are denoted by black, dashed dots, respectively, which act on adjacent qubits.
The green box depicts a \textit{square} in the toric code, which will be subdivided.
The collection of red boxes denotes a $Z$-type logical operator $\ell^{A}_{1}$.
(b) zooms in on a square and depicts the $(L=2)$-subdivision. 
The corners represent the original basis elements in square complex $A$.
The colored arrows denote gluing maps, which are elaborated in Propositions \eqref{prop:square-2}, \eqref{prop:square-1} and Theorem \eqref{thm:square}.
Compare with Diagram \eqref{eq:square-diagram}.
(c) depicts the subdivision of toric code in (a), with dashed guidelines removed; instead, consistent with (b), the blue squares denotes $A^{20}$, while the grey and blue lines denote $A^{10},A^{21}$, respectively.
The dashed red box/lines denote $\bra \ell_1^A\ket^1$ which is a representation of the logical operator induced by that in (a) via Proposition \eqref{thm:height-2}, while the solid red lines $\ell_1$ denote a more general representation. See Example \eqref{ex:square-logical}.
The two detached areas denote possible $\ell_2^2$ elements to \textit{clean} the general logical operator $\ell_1$ so that its projection onto $C^2_1$ is smaller.
}
\label{fig:square}
\end{figure}

\begin{prop}[Level 2]
    \label{prop:square-2}
    Define the sequence $C^{2}$ by the gluing procedure $A^{2} \xRightarrow{g_{21,2}} A^{21} \xRightarrow{g_{20,21}} A^{20}$ where
    \begin{align}
        A^{2}   &\equiv A_2 \to 0 \to 0\\
        A^{21}  &\equiv \bigoplus_{(a_2,a_1):a_1 \sim^{A} a_2} R^{\multimap} \to 0\\
        A^{20}  &\equiv \bigoplus_{(a_2,a_0):a_0 \sim^{A} a_2} R^{\multimap} \otimes R^{\multimap}
    \end{align}
    and the gluing maps $g_{\alpha,\beta}: A^{\beta} \to A^{\alpha}$ (blue arrows in Fig. \ref{fig:square-glue} or Diagram \eqref{eq:square-diagram}) are defined via
    \begin{align}
        g_{21,2} a_2 &= \sum_{a_1: a_1 \sim^{A} a_2} |1;a_2 a_1\ket\\
        g_{20,21} |s;a_2 a_1\ket &= \sum_{a_0: a_0 \sim^{A} a_2} (1\{a_1 =a_1^{\mathrm{h}}\} |s1;a_2 a_0\ket \nonumber\\
        &\quad + 1\{a_1 =a_1^{\mathrm{v}}\}|1s;a_2 a_0\ket)
    \end{align}
    where $a_1^{\mathrm{h}},a_1^{\mathrm{v}}$ are horizontal, vertical with respect to pair $(a_2,a_0)$, respectively, and $s$ is an integer or half-integer.
    Then $C^{2}$ (with $\partial^2$) is a regular cone and $H_2(\partial^{2}) \cong A_2$ where the isomorphism $A_2 \to H_2(\partial^{2})$ is given by $a_2 \mapsto [\bra a_2\ket^2]$ where
    \begin{align}
        \bra a_2\ket^2 &\equiv 
        \begin{pmatrix*}[l] 
        \phantom{\bra}a_2 \\ \bra a_2\ket^{21} \\ \bra a_2\ket^{20} 
        \end{pmatrix*} \\
        \bra a_2\ket^{21} &= \sum_{a_1:a_2\sim^{A} a_2} \sum_{i=1}^{L} |i^+;a_2 a_1\ket \\
        \bra a_2\ket^{20} &= \sum_{a_0:a_0\sim^{A} a_2} \sum_{i,j=1}^{L} |i^+j^+;a_2 a_0\ket
    \end{align}
    As in Fig. \ref{fig:square-subdivision}, $\bra a_2\ket^{21},\bra a_2\ket^{20}$ denote the collection of all blue branches, squares adjacent to $a_2$, respectively.
\end{prop}
\begin{proof}
    Note that the gluing procedure can be expanded into the following diagram, where omitted terms are 0.
    \begin{equation}
    \label{eq:square-2-diagram}
    \begin{tikzpicture}[baseline]
    \matrix(a)[matrix of math nodes, nodes in empty cells, nodes={minimum size=5mm},
    row sep=2em, column sep=2em,
    text height=1.5ex, text depth=0.25ex]
    {&& \cdot\\
    & \cdot  &  \cdot\\
    \cdot  & \cdot & \cdot\\};
    \path[->,font=\scriptsize]
    (a-2-2) edge  (a-2-3)
    (a-3-1) edge  (a-3-2)
    (a-3-2) edge  (a-3-3)
    (a-1-3) edge node[right]{$g_{21,2}$} (a-2-3)
    (a-2-2) edge (a-3-2)
    (a-2-3) edge node[right]{$g_{20,21}$} (a-3-3);
    \end{tikzpicture}
    \end{equation}
    One can then check that the lower-triangular matrix $\partial^{2}$ as defined by the gluing procedure for $C^{2}$ satisfies $\partial^{2}\partial^{2} =0$.
    By Lemma \eqref{lem:rep}, we see that $H_1(R^{\multimap}) =H_0(R^{\multimap})=0$ and thus we can apply Theorem \eqref{thm:height-n} for $n=2$ and $m=2$ so that the statement follows.
\end{proof}

\begin{prop}[Level 1]
    \label{prop:square-1}
    Define the sequence $C^{1}$ by the gluing procedure $A^{1} \xRightarrow{g_{10,1}} A^{10}$ where
    \begin{align}
        A^{1}   &\equiv A_1\to 0\\
        A^{10}  &\equiv \bigoplus_{(a_1,a_0):a_1 \sim^{A} a_0} R^{\multimap} 
    \end{align}
    and the gluing maps $g_{\alpha,\beta}: A^{\beta} \to A^{\alpha}$ (grey arrows in Fig. \ref{fig:square-glue} or Diagram \eqref{eq:square-diagram}) is defined via
    \begin{align}
        g_{10,1}a_1 &= \sum_{a_0:a_0\sim^{A} a_1} |1;a_1a_0\ket
    \end{align}
    Then $C^{1}$ (with $\partial^{1}$) is a regular cone and the isomorphism $A_1 \to H_1(\partial^{1})$ is given by $a_1 \mapsto [\bra a_1\ket^1]$ where
    \begin{align}
        \bra a_1\ket^{1}&=
        \begin{pmatrix*}[l] 
        \phantom{\bra}0 \\ \phantom{\bra}a_1 \\ \bra a_1\ket^{10} 
        \end{pmatrix*} \\
        \bra a_1\ket^{10} &= \sum_{a_0:a_0\sim^{A} a_1} \sum_{i=1}^{L} |i^+;a_1 a_0\ket
    \end{align}
    As shown in Fig. \ref{fig:square-subdivision}, $\bra a_1\ket^{10}$ denotes the collection of all grey branches adjacent to $a_1$.
\end{prop}
\begin{proof}
    Note that the gluing procedure can be expanded into the following diagram, where omitted terms are 0.
    \begin{equation}
    \label{eq:square-1-diagram}
    \begin{tikzpicture}[baseline]
    \matrix(a)[matrix of math nodes, nodes in empty cells, nodes={minimum size=5mm},
    row sep=2em, column sep=2em,
    text height=1.5ex, text depth=0.25ex]
    {& \cdot\\
     \cdot  &  \cdot\\};
    \path[->,font=\scriptsize]
    (a-2-1) edge  (a-2-2)
    (a-1-2) edge node[right]{$g_{10,1}$} (a-2-2);
    \end{tikzpicture}
    \end{equation}
    One can then check that the lower-triangular matrix $\partial^{1}$ as defined by the gluing procedure for $C^{1}$ satisfies $\partial^{1}\partial^{1} =0$.
    By Lemma \eqref{lem:rep}, we see that $H_1(R^{\multimap}) =H_0(R^{\multimap})=0$ and thus we can apply Theorem \eqref{thm:height-n} for $n=2$ and $m=1$ so that the statement follows.
\end{proof}

\begin{theorem}[$L$-subdivision]
    \label{thm:square}
    Define $C^{2},C^{1}$ as in Propositions \eqref{prop:square-2}, \eqref{prop:square-1} for fixed $L$ and define $C^{0}\equiv A^{0} \equiv (0 \to 0 \to A_0)$.
    Define the sequence $C$ obtained by 
    \begin{equation}
        C^{2} \xRightarrow{g_{1,21} \oplus g_{10,20}} C^{1} \xRightarrow{g_{0,10}} C^{0}
    \end{equation}
    where the gluing maps $g_{\alpha,\beta}:A^{\beta} \to A^{\alpha}$ (purple arrows in Fig. \ref{fig:square-glue} and Diagram \eqref{eq:square-diagram}) are defined via
    \begin{align}
        g_{1,21} |L^+;a_2 a_1\ket &= a_1 \\
        g_{10,20} |L^+, s;a_2 a_0\ket &= |s;a_1^{\mathrm{h}}a_0\ket \\
        g_{10,20} |s,L^+;a_2 a_0\ket &= |s;a_1^{\mathrm{v}}a_0\ket \\
        g_{0,10} |L^+; a_1a_0\ket  &= a_0
    \end{align}
    where $a_1^{\mathrm{h}},a_1^{\mathrm{v}}$ are horizontal, vertical with respect to pair $(a_2,a_0)$, respectively, and $s$ is an integer or half-integer.
    Then $C$ is a regular cone with embedded code $A$.
    We refer to $C$ as the \textbf{$L$-subdivision} of $A$.
\end{theorem}
\begin{proof}
    Note that the gluing procedure can be expanded into the following diagram, where omitted terms are 0.
    \begin{equation}
    \label{eq:square-ZQX-diagram}
    \begin{tikzpicture}[baseline]
    \matrix(a)[matrix of math nodes, nodes in empty cells, nodes={minimum size=5mm},
    row sep=2em, column sep=2em,
    text height=1.5ex, text depth=0.25ex]
    {\cdot &\cdot & \cdot \\
     \cdot &\cdot &  \\
     \cdot &      &  \\};
    \path[->,font=\scriptsize]
    (a-1-1) edge node[above]{$\partial^{2}$} (a-1-2)
    (a-1-2) edge node[above]{$\partial^{2}$} (a-1-3)
    (a-2-1) edge node[above]{$\partial^{1}$} (a-2-2)
    (a-1-1) edge node[left]{$g_2$} (a-2-1)
    (a-1-2) edge (a-2-2)
    (a-2-1) edge node[left]{$g_1$} (a-3-1);
    \end{tikzpicture}
    \end{equation}
    where, for notation simplicity, wrote $g_2 = g_{1,21} \oplus g_{10,20}$ and $g_1 = g_{0,10}$, and $\partial^{1},\partial^{0}$ are the differentials of $C^{1},C^{0}$, respectively.
    To show that the resulting $C$ is a complex, i.e., $\partial \partial =0$, it's sufficient to show that $g_2 \partial^{2} = \partial^{1} g_2$ (the square in the previous is commuting) and $g_1 g_2 =0$ (since $p=0$).
    Note that if $\partial^{\alpha}$ denotes the differential for $A^{\alpha}$ for all $\alpha$ except $\alpha=2,1,0$, then
    \begin{align}
        \partial^{2} &= \partial^{{20}} + \partial^{{21}} +g_{21,2} +g_{20,21} \\
        \partial^{1} &= \partial^{{10}} +g_{10,1}
    \end{align}
    And thus $g_2 \partial^{2} = \partial^{1} g_2$ is equivalent to
    \begin{align}
        g_{1,21} g_{21,2} &= 0 \\
        g_{1,21} \partial^{{21}} &= 0 \\
        g_{10,20} \partial^{{20}} &= \partial^{{10}} g_{10,20} \\
        g_{10,20}g_{20,21} &= g_{10,1} g_{1,21}
    \end{align}
    Similarly, $g_1 g_2 =0$ is equivalent to 
    \begin{equation}
        g_{0,10} g_{10,20} =0
    \end{equation}
    Hence, one can check that the equalities are all satisfied, i.e., Diagram \eqref{eq:square-diagram} is commutative. 
    Hence, the resulting $C$ is a complex and thus a regular cone by the previous propositions.
    In particular, we have the isomorphism between the embedded code and the cone $C$ as specified in Theorem \eqref{thm:height-2}.
    Therefore, all that's left is to show that the embedded code is given square complex $A$.
    Indeed, note that by Proposition \eqref{prop:square-2}, \eqref{prop:square-1}, $A_i \mapsto H_i(\partial^{i})$ via the isomorphism $a_i \mapsto \bra a_i\ket^{i}$ for $i=1,2$.
    Note that
    \begin{align}
        [g_2] [\bra a_2\ket^{2}] &= [g_2 \bra a_2\ket^{2}] \\
        &= [g_{1,21} \bra a_2\ket^{21} +g_{10,20} \bra a_2\ket^{20}] \\
        &= \sum_{a_1: a_1 \sim^{A} a_2} [a_1 +\bra a_1\ket^{10}] \\
        &= \sum_{a_1: a_1 \sim^{A} a_2} [\bra a_1\ket^{1}]
    \end{align}
    Hence, $[g_2] =\partial^{A}$.
    The case is similar for $[g_1]$ and thus we see that the embedded code is exactly $A$.
\end{proof}

\begin{remark}[3D Commutative Diagram]
Similar to Theorem \eqref{thm:height-2} and Diagram \eqref{eq:height-2-diagram}, Propositions \eqref{prop:square-2}, \eqref{prop:square-1} and Theorem \eqref{thm:square} can be summarized via the following (3-dimensional) commutative diagram
\begin{equation}
    \label{eq:square-diagram}
    \begin{tikzpicture}[baseline]
    \matrix(a)[matrix of math nodes, nodes in empty cells, nodes={minimum size=5mm},
    row sep=.6em, column sep=.5em,
    text height=1.5ex, text depth=0.25ex]
    {
        &&&&&& A_2 \\
        &&& A_2^{21} && A_1^{21} &\\
        A_2^{20} && A_1^{20} && A_0^{20} &&\\
        &&&A_1 &&& \\
        A_1^{10} && A_0^{10} &&&& \\
        &&&&&& \\
        A_0 &&&&&& \\
    };
    \path[->,font=\scriptsize]
    (a-1-1.center) edge node[right, xshift=2pt]{$x$} (a-1-1.east)
    (a-1-1.center) edge node[below left, xshift=-1pt]{$y$} (a-1-1.south west)
    (a-1-1.center) edge node[below, yshift=-2pt]{$z$} ([yshift=-2pt]a-1-1.south);
    \path[->,line width=.5mm,violet,font=\scriptsize]
    (a-3-1) edge (a-5-1)
    (a-3-3) edge (a-5-3)
    (a-5-1) edge (a-7-1);
    \path[->,font=\scriptsize]
    (a-2-4) edge (a-2-6)
    (a-3-1) edge (a-3-3)    
    (a-3-3) edge (a-3-5)
    (a-5-1) edge (a-5-3);
    \path[- ,line width=.5mm,violet,font=\scriptsize](a-2-4) edge ([yshift=4pt]a-3-4.center);
    \path[->,line width=.5mm,violet,font=\scriptsize]([yshift=-4pt]a-3-4.center) edge (a-4-4);
    \path[->,line width=.5mm,blue,font=\scriptsize]
    ([xshift=1pt,yshift=2pt]a-2-4.south west) edge (a-3-3)
    ([xshift=1pt,yshift=2pt]a-2-6.south west) edge (a-3-5)
    ([xshift=1pt,yshift=2pt]a-1-7.south west) edge (a-2-6);
    \path[->,line width=.5mm,gray,font=\scriptsize]
    ([xshift=1pt,yshift=2pt]a-4-4.south west) edge (a-5-3);
    \end{tikzpicture}  
\end{equation}
Where $0$s are omitted.
In correspondence with Fig. \ref{fig:square-glue}, the blue and grey arrows denoting gluing within levels $C^{2},C^{1}$, respectively, and the purple arrows denote gluing the levels $C^{2},C^{1},C^{0}$ together.
The comparison with Diagrams \eqref{eq:square-2-diagram}, \eqref{eq:square-1-diagram} and \eqref{eq:square-ZQX-diagram} is clear.
In particular, note that $C_{i}$ in the $L$-subdivision is the direct sum of all $\dF_2$ vectors spaces in the above diagram with subscript $i$.
\end{remark}

\begin{remark}[Symmetry of Construction]
\label{rem:square-symmetry}
From Fig. \ref{fig:square-glue}, the reader may notice that the levels are constructed differently and in an asymmetric manner.
This naturally prompts the question of how the $X$-type logical operators -- or, equivalently, the cochain complex -- are structured.
One (relatively tedious) way is to examine the full decomposition of $C$, i.e., $C_i = \bigoplus_{\alpha} A^{\alpha}_{i}$, and compare the codifferential $\partial^T$ with $\partial$. 
Alternatively, but possibly more apparent, one can examine Diagram \eqref{eq:square-diagram}.
Since the cochain complex is induced by inverting all arrows, we see that the resulting diagram has the same structure as Diagram \eqref{eq:square-diagram} after swapping the $y$- and $z$-axis ($y\leftrightarrow z$) and flipping the $x$-axis ($x \mapsto -x$).
\end{remark}

\subsubsection{Code Distance}

The next goal is to show that the mapping cone preserves code distance in some manner. 
To obtain some intuition, let us first consider the following example.

\begin{example}[Logical Operators]
\label{ex:square-logical}
Given Theorem \eqref{thm:square} or Diagram \eqref{eq:square-diagram}, the isomorphism between $H_1(A)$ and its subdivision $H_1(C)$ can be explicitly obtained via Theorem \eqref{thm:height-2}.
Here, we provide an intuitive description of the isomorphism by using the example of the toric code in Fig. \ref{fig:square}.
Fix a ($Z$-type) logical operator $\ell^{A}_{1} \in A_{1}$ regarded as a subset of $\sB(A_{1})$.
By the isomorphism $A_1 \cong H_1(C^{1})$, $\ell^{A}_{1}$ is mapped to some $[\ell^{1}_{1}] \in H_1(C^{1})$, where $[\ell^{1}_{1}]$ is the summation of $[\bra a_1\ket^1] \in H_1(C^{1})$ over $a_1 \in \ell^{A}_{1}$.
Choose a further representation  $\ell^{1}_{1} \in C^{1}_{1}$ of $[\ell^{1}_{1}]$.
By construction, $\bra \ell^{A}_{1}\ket^1$ is a possible representation of $[\ell^{1}_{1}]$ where $\bra \ell^{A}_1\ket^1$ is the summation of $\bra a_1\ket^1$ for all $a_1\in \ell_1^A$.
Since the isomorphism $H_1(A) \to H_1(C)$ is given by inclusion, i.e., $\llb \ell_1^1\rrb \mapsto [\ell_1^1]$, we see that $\bra \ell^{A}_{1}\ket^1$ is a logical operator in $C$, and shown via the dashed red boxes/lines in Fig. \ref{fig:square-subdivision}.
More generally, however, a logical operator could be of the form $\bra \ell_1^{A}\ket^1 + \partial \ell_2$ where $\ell_2\in C_2$, as shown by the solid red line in Fig. \ref{fig:square-subdivision}.
\end{example}

To prove a lower bound on the code distance of the $L$-subdivision $C$, the strategy is to use the Cleaning Lemma \eqref{lem:cleaning}, so that a general nontrivial logical operator (solid red line in Fig. \ref{fig:square-subdivision}) is \textit{cleaned} to a simpler form (dashed lines in Fig. \ref{fig:square-subdivision}), whose weight can be easily obtained. 
With that said, proving the isoperimetric inequality for a general square complex is somewhat difficult.
Instead, Ref. \cite{lin2023geometrically} proves the isoperimetric inequality for the particular case where $C^{2}$ has a nice product structure.
Specifically, 
\begin{lemma}[Lemma 5.2 of Ref. \cite{lin2023geometrically}]
    \label{lem:square-isoperimetric}
    Let $A$ be a square complex with basis. Further assume that for every $a_2$, it is possible to partition support $\supp^{A}_1 a_2$ into $\sA_1^{\mathrm{h}}=\sA_1^{\mathrm{h}}(a_2),\sA_1^{\mathrm{v}}=\sA_1^{\mathrm{v}}(a_2)$ so that for every $a_0 \sim^{A} a_2$, there exists a unique pair $(a_1^{\mathrm{h}},a_1^{\mathrm{v}})\in \sA_1^{\mathrm{h}}\times \sA_1^{\mathrm{v}}$ such that $a_1^{\mathrm{h}},a_1^{\mathrm{v}} \sim^{A} a_0$ and vice-versa.
    Then $C^{2}$ is the direct sum of $C^{2,a_2}$ where
    \begin{align}
        C^{2,a_2} &\cong \left((a_2 \to 0)\Rightarrow \bigoplus_{a_1\in \sA_1^{\mathrm{h}}} R^{\multimap}\right) \nonumber\\
        &\quad\otimes \left((a_2 \to 0)\Rightarrow \bigoplus_{a_1\in \sA_1^{\mathrm{v}}(a_2)} R^{\multimap}\right)
    \end{align}
    And that $(\partial^{2,a_2},g_2)$ satisfies the isoperimetric inequality with coefficient $\Theta(1)$ for all $a_2$.
\end{lemma}

\begin{theorem}[Subdivision Distance]
    \label{thm:square-distance}
    Assume the same hypothesis as in Theorem \eqref{thm:square} so that $A$ is the embedded code within the $L$-subdivision $C\equiv C(L)$.
    Further assume the same hypothesis as in Lemma \eqref{lem:square-isoperimetric} and that with $a_2 \leftrightarrow a_0$ exchanged.
    Let $d_\alpha^{A}$ be the $\alpha=Z,X$-type code distance of $A$, respectively.
    Then the $Z$- and $X$-type code distance $d_Z,d_X$ of the regular height-2 cone $C$ is bounded below by
    \begin{equation}
        d_\alpha \ge \Theta(1) L d^{A}_\alpha , \quad \alpha =Z,X
    \end{equation}
    Where $\Theta(1)$ is some constant number.
\end{theorem}

\begin{proof}
    By symmetry of the construction -- Remark \eqref{rem:square-symmetry}, it's sufficient to prove the statement for the $Z$-type code distance.
    By the Cleaning Lemma \eqref{lem:cleaning}, we have
    \begin{equation}
        d_Z \ge \Theta(1) \min_{\ell_1^1\in C_1^1: 0\ne \llb \ell_1^1 \rrb \in H_1(A)} |\ell_1^1|
    \end{equation}
    By Proposition \eqref{prop:square-1}, we see that any nontrivial logical operator $\ell_1^A \in H_1(C^1)$ must have representation $\ell_1^1 \in C^1_1$ given by $\bra \ell_1^A\ket^1$, i.e.,
    \begin{equation}
        \ell_1^1 =\sum_{a_1\in \ell_1^A} \bra a_1\ket^1
    \end{equation}
    Therefore, the statement follows.
\end{proof}

\section{Quantum Weight Reduction}
\label{sec:weight}

In this section, we discuss how the framework fits into Hastings' work on quantum weight reduction \cite{hastings2016weight,hastings2021quantum} and its relation with fault-tolerant logical measurement \cite{williamson2024low,ide2025fault,cross2024improved,horsman2012surface}.
In particular, we elaborate on Ref. \cite{williamson2024low}, whose construction is similar to Hastings \cite{hastings2016weight} but applied in the context of logical measurement.
Specifically, given complex $A$ and logical operator $\ell_1^\star\in A_1$, an ancillary complex $B$ and gluing map $g$ is cleverly constructed so that the height-1 cone $C = B\Rightarrow A$, expanded as,
\begin{equation}
    \label{eq:logical-measurement-diagram}
    \begin{tikzpicture}[baseline]
    \matrix(a)[matrix of math nodes, nodes in empty cells, nodes={minimum size=20pt},
    row sep=1.5em, column sep=2em,
    text height=1.5ex, text depth=0.25ex]
    {
    & B_2 & B_1 & B_0\\
    A_2 & A_1 & A_0   \\};
    \path[->,font=\scriptsize]
    (a-1-2) edge (a-1-3)
    (a-1-3) edge (a-1-4)
    (a-2-1) edge (a-2-2)
    (a-2-2) edge (a-2-3)
    (a-1-2) edge node[right]{$g$} (a-2-2) 
    (a-1-3) edge node[right]{$g$} (a-2-3);
    \end{tikzpicture}
    \end{equation}
is regular (with respect to degree 1) and $\im [g] = [\ell_1^\star]$.
By Theorem \eqref{thm:height-n} for $n=1$, $C$ has all the logical operators of $A$ except $\ell_1^\star$, i.e., $H_1(C) \cong H_1(A)/[\ell_1^\star]$, and $\ell_1^\star$ can be induced by local parity checks in $C$.
The goal of this section is then to construct $B$ and $g$ in an algebraic manner.

\subsection{Ancillary Code}

Given a connected graph $\sG= (\sV,\sE)$ with vertices $\sV$ and edges $\sE$, let $V,E$ denote the $\dF_2$ vector spaces generated by $\sV,\sE$, respectively, and define the differential $E\to V$ by the adjacency of vertices and edges in graph $\sG$. 
Moreover, by the following lemma
\begin{lemma}[Decongestion \cite{hastings2021quantum,freedman2021building}]
    \label{lem:decongestion}
    Let $\sG=(\sE,\sV)$ be an arbitrary graph. Then there exists a basis $\sF$ of simple cycles in $\sG$ which has total weight -- number of edges -- $O(|\sV|\log |\sV|)$, and that each edge appears in the basis at most $O(\log^2 |\sV|)$ times.
\end{lemma}
We can choose basis $\sF$, let $F$ denote the $\dF_2$ vector space generated by $\sF$, and define the differential $F\to E$ by the adjacency of simple cycles and edges.
In particular, $G\equiv F\to E \to V$ is a complex with trivial 1-(co)homology, i.e., $H_1(G)=H_1(G^T)=0$, and 0-cohomology $H_0(G^T)$ spanned by a unique basis element corresponding to the connected component of $\sG$, i.e.,
\begin{equation}
    [\sV] \equiv \left[\sum_{v\in \sV} v\right]
\end{equation}
It then follows that
\begin{prop}[Section III of Ref. \cite{hastings2021quantum} or Lemma 1 of Ref. \cite{williamson2024low}]
    \label{prop:weight-cone}
    Let complex $A$ of CSS code be given with basis $\sB(A)$ and choose logical operator $\ell_1^\star\in A_1$, regarded as a subset of $\sB(A_1)$.
    Let $\sG$ be a connected graph such that there exists bijection $g:\sV \to \ell_1^\star$ (which extends to a linear $B_2 \to A_1$), and let $B\equiv G^T$.
    Let $C$ denote the sequence obtained from Diagram \eqref{eq:logical-measurement-diagram}, where $g:B_1 \to A_0$ is defined as follows.
    
    For every $a_0$ adjacent to some $a_1\in \ell_1^\star$, we see that $\supp^A a_0 \cap \ell_1^\star$ must have even cardinality, and thus there exists edge-paths $\gamma$ in $\sG$ connecting vertices in $\sG$ corresponding to $\supp^A a_0 \cap \ell_1^\star$ via the bijection $g:\sV \to \ell_1^\star$. Let $\Gamma (a_0)$ denote the union of such edge paths and define $g:B_1 \to A_0$ via
    \begin{equation}
        g e = \sum_{a_0} a_0 1\{e\in \Gamma(a_0)\} 
    \end{equation}
    Then $C$ is a regular height-1 cone with $H_1(C) \cong H_1(A)/[\ell_1^\star]$.
\end{prop}
\begin{proof}
    It's straightforward to check that $g$ is a chain map, i.e., $g \partial^B = \partial^{A} g$. Also note that
    \begin{equation}
        [g][\sV]= [\ell_1^A]
    \end{equation}
    And thus the statement follows from Theorem \eqref{thm:height-n} for $n=1$.
\end{proof}

A version of the previous proposition first appeared in Ref. \cite{hastings2021quantum}, where $\ell_1^\star$ was chosen to be a trivial logical operator ($=\partial^{A} a_2^\star$ for some $a_2^\star$) with large weight that we wish to reduce, and the edges of the graph $\sG$ was chosen to be a perfect matching of $\sV$.
Ref. \cite{williamson2024low} then generalized the statement for arbitrary graphs $\sG$, and chose $\ell_1^\star$  to be a nontrivial logical operator that we wish to measure.

The following proposition then guarantees that the height-1 cone $C$ has code distance that is lower bounded by that of the original code $A$ up to some constant.
\begin{prop}[Lemma 8 of Ref. \cite{hastings2021quantum} or Lemma 2 of Ref. \cite{williamson2024low}]
    \label{prop:weight-cone-distance}
    Assuming the same hypothesis as in Proposition \eqref{prop:weight-cone}, let $h(\sG)$ denote the Cheeger constant of $\sG$. Then the ($Z$-type) code distances $d_Z^C,d_Z^A$ of $C,A$ are related via
    \begin{equation}
        d^C_Z \ge \min(h(\sG),1) d^A_Z
    \end{equation}
\end{prop}
\begin{proof}
By definition of the Cheeger constant, we see that
\begin{equation}
    |\partial^B \ell_2^B| \ge h(\sG) \min(|g\ell_2^B|,|g(\ell_2^B +\sV)|)
\end{equation}
where $\sV$ is regarded as an element in $B_2$ via the natural identification (see Definition \eqref{def:basis}) and thus satisfies the isoperimetric inequality.
By the Cleaning Lemma \eqref{lem:cleaning}, the statement then follows.
\end{proof}

\subsection{Triangulation and Thickening}
Despite the simplicity of the construction in the previous section, the ancilla code $B\equiv G^T$ is not yet sufficient.
Indeed, the final code $C$ should be LDPC and thus in particular, $B$ must also be LDPC.
However, so far, code $G$ has weights
\begin{equation}
    F \xrightleftharpoons[O(\log^2|\sV|)]{\max_{f\in\sF} |f|} E \xrightleftharpoons[2]{\Delta(\sG)} V
\end{equation}

To reduce $q_X(G)$, we can choose a graph $\sG$ so that the maximum vertex degree $\Delta(\sG)=O(1)$.

\begin{figure}[ht]
\centering
\includegraphics[width=0.4\columnwidth]{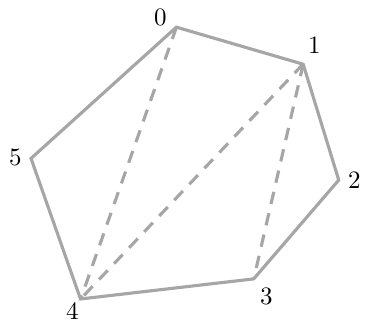}
\caption{Triangulation. The solid lines indicate the original simple cycle $f\in\sF$, while the dashed lines indicate the added edges.
}
\label{fig:weight-triangulate}
\end{figure}

To reduce $w_Z(G)$, Ref. \cite{hastings2021quantum,williamson2024low} \textit{triangulated} each cycle $f\in \sF$ as shown in Fig. \ref{fig:weight-triangulate}. 
Specifically, let $f$ have vertices labeled in sequence via $0,1,...,|f|-1$. 
Add edges $(0,|f|-2),(|f|-2,1),(1,|f|-3),...$ to $\sE$ and replace $f\in \sF$ with cycles $(0,|f|-2,|f|-1),(0,|f|-2,1),...$.
Perform the operation for every $f\in \sF$ to obtain $\sV^\tri =\sV,\sE^{\tri},\sF^{\tri}$ and correspondingly, $G^{\tri}:F^{\tri}\to E^{\tri} \to V^{\tri}$.
By Theorem \eqref{thm:height-2} (compare with Theorem \eqref{thm:barycentric-subdivision} and Fig. \ref{fig:barycenter}), it's straightforward to check that $G^\tri$ is a regular height-2 cone with embedded code $G$, but with weights
\begin{equation}
    F^\tri \xrightleftharpoons[O(\log^2|\sV|)]{3} E^\tri \xrightleftharpoons[2]{O(\Delta(\sG))} V^{\tri}
\end{equation}

To reduce $q_Z(G)$, Ref. \cite{hastings2021quantum,williamson2024low} performed a \textit{thickening} procedure summarized as follows.
\begin{prop}[Lemma 2 of Ref. \cite{hastings2021quantum} or Definition 3 of Ref. \cite{williamson2024low}]
    \label{prop:weight-thicken}
    Given an arbitrary complex $G=F\to E\to V$ with basis $\sB(G)$.
    Let $L\ge 3$ and $h:\sB(F) \to \{1,...,L\}$ be an injective map and let $G^{\mathrm{thick}}$ be obtained via the following gluing procedure
    \begin{equation}
    \begin{tikzpicture}[>=implies,baseline]
    \matrix(a)[matrix of math nodes, nodes in empty cells, nodes={minimum size=20pt},
    row sep=1.5em, column sep=2em,
    text height=1.5ex, text depth=0.25ex]
    {
    \quad\quad(F \to 0 \to 0) \\
    (E \otimes R(L) \to 0) \\
    (0 \to V\otimes R(L))\quad\quad\\};
    \draw[double,->,font=\scriptsize](a-1-1) -- node[right]{$g_2$} (a-2-1);
    \draw[double,->,font=\scriptsize](a-2-1) -- node[right]{$g_1$} (a-3-1);
    \end{tikzpicture}
    \end{equation}
    where
    \begin{align}
        g_2 f &= \partial^G f\otimes |h(f)\ket \\
        g_1 e\otimes |s\ket &= \partial^G e \otimes |s\ket
    \end{align}
    where $s$ is an integer or half-integer, $i$ is an integer and $f\in \sB(F),e\in \sB(E)$.
    Then $G^{\mathrm{thick}}$ is a regular (height-2) cone with embedded code $G$, referred as the \textbf{$Z$-thickening code} of $G$ with \textbf{length} $L$ and \textbf{height function} $h$.
    Moreover, the isomorphism $H^0(G) \to H^0(G^\mathrm{thick})$ is induced by inclusion, i.e.,
    \begin{equation}
        [\ell_0^G] \mapsto [\ell_0^G]\otimes [\sR_0]
    \end{equation}
    where $\sR_0$ is defined in Lemma \eqref{lem:rep}.
\end{prop}

\begin{figure}[ht]
\centering
\subfloat[\label{fig:weight-ancilla}]{%
    \centering
    \includegraphics[width=0.6\columnwidth]{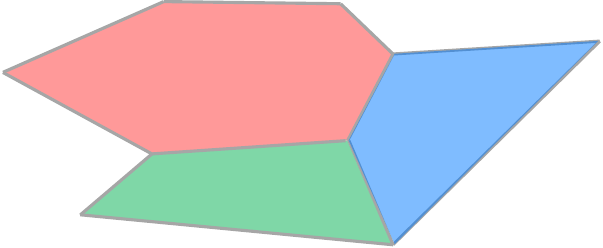}
}
\\
\subfloat[\label{fig:weight-copy}]{%
    \centering
    \includegraphics[width=0.6\columnwidth]{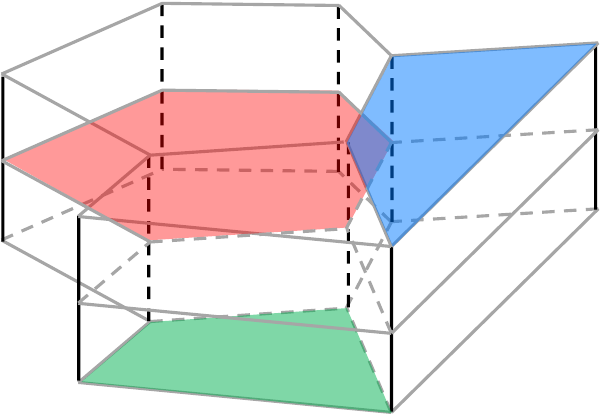}
}
\caption{Thickening. (a) depicts the complex $G$ corresponding to a graph $\sG$ with generating simple cycles $\sF$ colored. 
(b) depicts the thickened $G^{\mathrm{thick}}$ for $L=3$, where the original cycles $f\in \sF$ are mapped to different heights via the height function $h$.
}
\label{fig:weight-thicken}
\end{figure}
The proof follows straightforwardly via the main result in Theorem \eqref{thm:height-2} and \eqref{thm:height-n}, and thus it may be more beneficial to understand the intuition behind the thickening construction.
Indeed, expand the cone construction as the following diagram (where 0s and $\otimes$ signs are omitted)
\begin{equation}
\label{eq:weight-thicken-diagram}
\begin{tikzpicture}[baseline]
\matrix(a)[matrix of math nodes, nodes in empty cells, nodes={minimum size=20pt},
row sep=2em, column sep=2em,
text height=1.5ex, text depth=0.25ex]
{
& F & \\
E R_{1}(L) & E R_{0}(L) &   \\
V R_{1}(L) & V R_{0}(L) & \\};
\path[->,font=\scriptsize]
(a-1-2) edge node[right]{$g_2$} (a-2-2)
(a-2-2) edge node[right]{$g_1$} (a-3-2)
(a-2-1) edge (a-3-1)
(a-2-1) edge (a-2-2)
(a-3-1) edge (a-3-2);
\end{tikzpicture}
\end{equation}
Although the thickening procedure holds for all CSS codes, restrict our attention to the case where $G$ is constructed from the graph $\sG$ and generating simple cycles $\sF$ as shown in Fig. \ref{fig:weight-ancilla}.
It's then clear that Fig. \ref{fig:weight-copy} depicts the complex $G^{\mathrm{thick}}$ for $L=3$, which has generating cycles depicted by the colored plaquettes and all vertical plaquettes.
Note that in Fig. \ref{fig:weight-ancilla}, there exist an edge adjacent to both red and green plaquettes.
In contrast, any edge in Fig. \ref{fig:weight-copy} can only be adjacent to either the red or the green plaquette, and thus $q_Z(G)$ can be reduced by a clever choice of the height function.
Specifically,

\begin{prop}[Lemma 5 of Ref. \cite{hastings2021quantum}]
    Let $G$ be the complex associated with graph $\sG$ and generating simple cycles $\sF$.
    Then there exists $L=O(\log^2 |\sV| +\Delta(\sG))$ and height function $h$ such that the thickened complex $G^{\mathrm{thick}}$ has weights satisfying
    \begin{equation}
        F^{\mathrm{thick}} \xrightleftharpoons[\max(\Delta(\sG),3)]{\max(\max_{f\in \sF} |f|,4)} E^{\mathrm{thick}}  \xrightleftharpoons[2]{\Delta(\sG)+2} V^{\mathrm{thick}} 
    \end{equation}
\end{prop}

\begin{proof}
    The upper bounds for $w_Z,w_X,q_X$ follows straightforwardly from the matrix representation of the height-2 cone $\partial^{\mathrm{thick}}$ (see Fig. \ref{fig:weight-thicken} for intuition), i.e.,
    \begin{align}
    \partial_{2}^{\mathrm{thick}} &=
    \begin{pmatrix}
        0   & 0 & 0 \\
        g_2 & \partial^{R} & 0 \\
        0   & g_1 & 0
    \end{pmatrix},
    \\
    \partial_{1}^{\mathrm{thick}} &=
    \begin{pmatrix}
        0   & 0 & 0 \\
        0   & 0 & 0 \\
        0   & g_1 & \partial^{R}
    \end{pmatrix}
    \end{align}
    Hence, we shall focus on $q_Z(G^{\mathrm{thick}})$.
    Note that
    \begin{align}
        q_X(B^{\mathrm{thick}}) &= |\partial_2^{\mathrm{thick}}|_{\mathrm{row}}\\
        &= \max\left(\Delta(\sG), 2+|g_2|_{\mathrm{row}}\right)
    \end{align}
    By definition, we see that
    \begin{equation}
        |g_2|_{\mathrm{row}} = \max_{e, i} \sum_{f \sim^{G} e} 1\{i=h(f)\}
    \end{equation}
    Consider the bipartite graph consisting of vertices $\sE \sqcup \sF)$ with edges $ef$ if $f \sim^{G} e$.
    Note that the maximum vertex degree of the graph is given by $\max(q_Z(G),\Delta(\sG))$. Hence, by the greedy coloring theorem, we see that if $L=\max(q_Z(G),\Delta(\sG)) +1$, there exists color/height function $h$ such that $|g_2|_{\mathrm{row}}=1$
\end{proof}

\subsection{Low Overhead Measurement}

By the previous section, choose connected graph $\sG$ of constant degree $\Delta(\sG)=O(1)$ with generating simple cycles $\sF$ as in the decongestion lemma \eqref{lem:decongestion} and let $G$ denote the induced complex.
Let $\hat{G}$ be the complex obtained after triangulating and thickening.
Then $\hat{G}^T$ is an LDPC code acting on $O(|\sV| \log^3 |\sV|)$ qubits.
It has trivial 1-(co)homology, and its 0-homology is spanned by the unique basis element
\begin{equation}
    [\sV] \otimes \left[\sR_0 \right]
\end{equation}
Using $\hat{G}^T$ as the ancilla code, we obtain the final result (whose proof is similar to Propositions \eqref{prop:weight-cone}-\eqref{prop:weight-cone-distance} and thus omitted).
\begin{theorem}
    Let complex $A$ of CSS code be given with basis and choose logical operator $\ell_1^\star\in A_1$, regarded as a subset of 1-cells.
    Let $\sG$ be a connected graph such that there exists bijection $g:\sV \to \ell_1^\star$  and $\hat{G}$ be defined with $B\equiv \hat{G}^T$.
    Let $\sG$ be of constant degree $\Delta(\sG)= O(1)$ and let each edge path $\gamma$ defined in Proposition \eqref{prop:weight-cone} be of $O(1)$ length.
    Define $g:B_i \to A_{i-1}$ via
    \begin{align}
        g v\otimes |i\ket &= 1\{i=0\} gv \\
        g e\otimes |i\ket &= 1\{i=0\} \sum_{a_0} a_0 1\{ e\in \Gamma(a_0)\} 
    \end{align}
    where $i$ is an integer.
    Then the sequence $C$ obtained from Diagram \eqref{eq:logical-measurement-diagram} with $B=\hat{G}^T$ is a regular height-1 cone with $H_1(C) \cong H_1(A)/[\ell_1^\star]$ and $O(1)$ weights and has $Z$-type code distance $d_Z$ satisfying
    \begin{equation}
        d_Z \ge \min(h(\sG),1) d_Z^{A}
    \end{equation}
    where $h(\sG)$ is the Cheeger constant of $\sG$ and $d_Z^A$ is the $Z$-type code distance of $A$
\end{theorem}


\bibliographystyle{alpha}
\bibliography{main.bbl}

\begin{thebibliography}{HFDVM12}

\bibitem[ABO97]{aharonov1997fault}
Dorit Aharonov and Michael Ben-Or.
\newblock Fault-tolerant quantum computation with constant error.
\newblock In {\em Proceedings of the twenty-ninth annual ACM symposium on
  Theory of computing}, pages 176--188, 1997.

\bibitem[BBC25]{baspin2025fast}
Nou{\'e}dyn Baspin, Lucas Berent, and Lawrence~Z Cohen.
\newblock Fast surgery for quantum ldpc codes.
\newblock {\em arXiv preprint arXiv:2510.04521}, 2025.

\bibitem[BE21]{breuckmann2021balanced}
Nikolas~P Breuckmann and Jens~N Eberhardt.
\newblock Balanced product quantum codes.
\newblock {\em IEEE Transactions on Information Theory}, 67(10):6653--6674,
  2021.

\bibitem[BK98]{bravyi1998quantum}
Sergey~B Bravyi and A~Yu Kitaev.
\newblock Quantum codes on a lattice with boundary.
\newblock {\em arXiv preprint quant-ph/9811052}, 1998.

\bibitem[BPT10]{bravyi2010tradeoffs}
Sergey Bravyi, David Poulin, and Barbara Terhal.
\newblock Tradeoffs for reliable quantum information storage in 2d systems.
\newblock {\em Physical review letters}, 104(5):050503, 2010.

\bibitem[BT09]{bravyi2009no}
Sergey Bravyi and Barbara Terhal.
\newblock A no-go theorem for a two-dimensional self-correcting quantum memory
  based on stabilizer codes.
\newblock {\em New Journal of Physics}, 11(4):043029, 2009.

\bibitem[BTL10]{bravyi2010majorana}
Sergey Bravyi, Barbara~M Terhal, and Bernhard Leemhuis.
\newblock Majorana fermion codes.
\newblock {\em New Journal of Physics}, 12(8):083039, 2010.

\bibitem[CHRY24]{cross2024improved}
Andrew Cross, Zhiyang He, Patrick Rall, and Theodore Yoder.
\newblock Improved qldpc surgery: Logical measurements and bridging codes.
\newblock {\em arXiv preprint arXiv:2407.18393}, 2024.

\bibitem[CHWY25]{cowtan2025fast}
Alexander Cowtan, Zhiyang He, Dominic~J Williamson, and Theodore~J Yoder.
\newblock Fast and fault-tolerant logical measurements: Auxiliary hypergraphs
  and transversal surgery.
\newblock {\em arXiv preprint arXiv:2510.14895}, 2025.

\bibitem[CRSS97]{calderbank1997quantum}
A~Robert Calderbank, Eric~M Rains, Peter~W Shor, and Neil~JA Sloane.
\newblock Quantum error correction and orthogonal geometry.
\newblock {\em Physical Review Letters}, 78(3):405, 1997.

\bibitem[DHLV23]{dinur2023good}
Irit Dinur, Min-Hsiu Hsieh, Ting-Chun Lin, and Thomas Vidick.
\newblock Good quantum ldpc codes with linear time decoders.
\newblock In {\em Proceedings of the 55th annual ACM symposium on theory of
  computing}, pages 905--918, 2023.

\bibitem[DKLP02]{dennis2002topological}
Eric Dennis, Alexei Kitaev, Andrew Landahl, and John Preskill.
\newblock Topological quantum memory.
\newblock {\em Journal of Mathematical Physics}, 43(9):4452--4505, 2002.

\bibitem[FH21]{freedman2021building}
Michael Freedman and Matthew Hastings.
\newblock Building manifolds from quantum codes.
\newblock {\em Geometric and Functional Analysis}, 31(4):855--894, 2021.

\bibitem[Got96]{gottesman1996class}
Daniel Gottesman.
\newblock Class of quantum error-correcting codes saturating the quantum
  hamming bound.
\newblock {\em Physical Review A}, 54(3):1862, 1996.

\bibitem[Got13]{gottesman2013fault}
Daniel Gottesman.
\newblock Fault-tolerant quantum computation with constant overhead.
\newblock {\em arXiv preprint arXiv:1310.2984}, 2013.

\bibitem[Haa13]{haah2013commuting}
Jeongwan Haah.
\newblock Commuting pauli hamiltonians as maps between free modules.
\newblock {\em Communications in Mathematical Physics}, 324:351--399, 2013.

\bibitem[Has16]{hastings2016weight}
Matthew~B Hastings.
\newblock Weight reduction for quantum codes.
\newblock {\em arXiv preprint arXiv:1611.03790}, 2016.

\bibitem[Has21]{hastings2021quantum}
Matthew~B Hastings.
\newblock On quantum weight reduction.
\newblock {\em arXiv preprint arXiv:2102.10030}, 2021.

\bibitem[Hat00]{hatcher2005algebraic}
Allen Hatcher.
\newblock {\em {Algebraic topology}}.
\newblock Cambridge Univ. Press, Cambridge, 2000.

\bibitem[HFDVM12]{horsman2012surface}
Dominic Horsman, Austin~G Fowler, Simon Devitt, and Rodney Van~Meter.
\newblock Surface code quantum computing by lattice surgery.
\newblock {\em New Journal of Physics}, 14(12):123011, 2012.

\bibitem[IGND25]{ide2025fault}
Benjamin Ide, Manoj~G Gowda, Priya~J Nadkarni, and Guillaume Dauphinais.
\newblock Fault-tolerant logical measurements via homological measurement.
\newblock {\em Physical Review X}, 15(2):021088, 2025.

\bibitem[Kit03]{kitaev2003fault}
A~Yu Kitaev.
\newblock Fault-tolerant quantum computation by anyons.
\newblock {\em Annals of physics}, 303(1):2--30, 2003.

\bibitem[KK12]{kitaev2012models}
Alexei Kitaev and Liang Kong.
\newblock Models for gapped boundaries and domain walls.
\newblock {\em Communications in Mathematical Physics}, 313(2):351--373, 2012.

\bibitem[KL96]{knill1996concatenated}
Emanuel Knill and Raymond Laflamme.
\newblock Concatenated quantum codes.
\newblock {\em arXiv preprint quant-ph/9608012}, 1996.

\bibitem[KLZ96]{knill1996threshold}
Emanuel Knill, Raymond Laflamme, and Wojciech Zurek.
\newblock Threshold accuracy for quantum computation.
\newblock {\em arXiv preprint quant-ph/9610011}, 1996.

\bibitem[KP12]{kovalev2012fault}
Alexey~A Kovalev and Leonid~P Pryadko.
\newblock Fault-tolerance of" bad" quantum low-density parity check codes.
\newblock {\em arXiv preprint arXiv:1208.2317}, 2012.

\bibitem[LWH23]{lin2023geometrically}
Ting-Chun Lin, Adam Wills, and Min-Hsiu Hsieh.
\newblock Geometrically local quantum and classical codes from subdivision.
\newblock {\em arXiv preprint arXiv:2309.16104}, 2023.

\bibitem[LZ22]{leverrier2022quantum}
Anthony Leverrier and Gilles Z{\'e}mor.
\newblock Quantum tanner codes.
\newblock In {\em 2022 IEEE 63rd Annual Symposium on Foundations of Computer
  Science (FOCS)}, pages 872--883. IEEE, 2022.

\bibitem[NC10]{nielsen2010quantum}
Michael~A Nielsen and Isaac~L Chuang.
\newblock {\em Quantum computation and quantum information}.
\newblock Cambridge university press, 2010.

\bibitem[PK22]{panteleev2022asymptotically}
Pavel Panteleev and Gleb Kalachev.
\newblock Asymptotically good quantum and locally testable classical ldpc
  codes.
\newblock In {\em Proceedings of the 54th Annual ACM SIGACT Symposium on Theory
  of Computing}, pages 375--388, 2022.

\bibitem[Por23]{portnoy2023local}
Elia Portnoy.
\newblock Local quantum codes from subdivided manifolds.
\newblock {\em arXiv preprint arXiv:2303.06755}, 2023.

\bibitem[QVRC21]{quintavalle2021single}
Armanda~O Quintavalle, Michael Vasmer, Joschka Roffe, and Earl~T Campbell.
\newblock Single-shot error correction of three-dimensional homological product
  codes.
\newblock {\em PRX Quantum}, 2(2):020340, 2021.

\bibitem[RH07]{raussendorf2007fault}
Robert Raussendorf and Jim Harrington.
\newblock Fault-tolerant quantum computation with high threshold in two
  dimensions.
\newblock {\em Physical review letters}, 98(19):190504, 2007.

\bibitem[Sho96]{shor1996fault}
Peter~W Shor.
\newblock Fault-tolerant quantum computation.
\newblock In {\em Proceedings of 37th conference on foundations of computer
  science}, pages 56--65. IEEE, 1996.

\bibitem[Til00]{tillich2000edge}
Jean-Pierre Tillich.
\newblock Edge isoperimetric inequalities for product graphs.
\newblock {\em Discrete Mathematics}, 213(1-3):291--320, 2000.

\bibitem[Tod62]{toda1962composition}
Hiroshi Toda.
\newblock {\em Composition methods in homotopy groups of spheres}.
\newblock Number~49. Princeton University Press, 1962.

\bibitem[TZ13]{tillich2013quantum}
Jean-Pierre Tillich and Gilles Z{\'e}mor.
\newblock Quantum ldpc codes with positive rate and minimum distance
  proportional to the square root of the blocklength.
\newblock {\em IEEE Transactions on Information Theory}, 60(2):1193--1202,
  2013.

\bibitem[WB24]{williamson2024layer}
Dominic~J Williamson and Nou{\'e}dyn Baspin.
\newblock Layer codes.
\newblock {\em Nature Communications}, 15(1):9528, 2024.

\bibitem[Wei94]{weibel1994introduction}
Charles~A Weibel.
\newblock {\em An introduction to homological algebra}.
\newblock Number~38. Cambridge university press, 1994.

\bibitem[WY24]{williamson2024low}
Dominic~J Williamson and Theodore~J Yoder.
\newblock Low-overhead fault-tolerant quantum computation by gauging logical
  operators.
\newblock {\em arXiv preprint arXiv:2410.02213}, 2024.

\bibitem[ZJX25]{zheng2025high}
Guo Zheng, Liang Jiang, and Qian Xu.
\newblock High-rate surgery: towards constant-overhead logical operations.
\newblock {\em arXiv preprint arXiv:2510.08523}, 2025.

\end{thebibliography}

\appendix
\onecolumngrid
\section{Proof of Main Results}


In standard homological algebra \cite{weibel1994introduction,hatcher2005algebraic}, given a chain map $f:C^{1}_{i} \to C^{0}_{i-1}$, i.e., $f\partial^{1}=\partial^{0}f$, the \textbf{cone}  is defined as the chain complex $\cone(f)$ with $i$ degree $C^{1}_{i} \oplus C^{0}_{i}$ and differential 
\begin{equation}
    \partial^{\cone(f)} =
    \begin{pmatrix}
    \partial^{1} &  \\
    f & \partial^{0} 
    \end{pmatrix}
\end{equation}


\begin{lemma}[Theorem 1.3.1 and Lemma 1.5.3 of Ref. \cite{weibel1994introduction}]
    \label{lem:exact-seq}
    Let $0 \to A \xrightarrow{f} B\xrightarrow{g}  C \to 0$ be a short exact sequence of chain complexes.
    Then there are natural maps $\partial: H_n(C) \to H_{n-1}(A)$, called \textbf{connecting homomorphisms}, such that the following is a long exact sequence
    \begin{equation}
        \cdots \to H_n(A) \xrightarrow{[f]}  H_n(B) \xrightarrow{[g]} H_n(C) \xrightarrow{\partial}  H_{n-1}(A) \to \cdots
    \end{equation}
    Where $[f],[g]$ are induced by $f,g$ via the quotient maps.
    
    In particular, if $f:C^{1}_{i}\to C^{0}_{i-1}$ is a chain map, then $0 \to C^{0} \xrightarrow{\iota} \cone(f) \xrightarrow{\pi} C^{1} \to 0$ is an exact sequence with inclusion $\iota$ and projection $\pi$, and the induced long exact sequence has natural connecting homomorphism $\partial = [f]$, i.e., that induced by $f$ on the homologies.
\end{lemma}


\begin{proof}[Proof of Theorem \eqref{thm:height-2}]
\label{proof:height-2}
Consider the following diagram,
\begin{equation}
\begin{tikzpicture}[baseline]
\matrix(a)[matrix of math nodes, nodes in empty cells, nodes={minimum size=5mm},
row sep=3em, column sep=3em,
text height=1.5ex, text depth=0.25ex]
{& H_1(\partial^{0})=0 &   & \\
H_2(\partial^{2}) &  H_1\begin{pmatrix} \partial^{1} & \\ g_1 & \partial^{0} \end{pmatrix} & H_1(\partial)  & H_1(\partial^{2})=0 \\
 & H_1(\partial^{1}) & &\\
 & H_0(\partial^{0}) & &\\};
\path[->,font=\scriptsize](a-2-1) edge node[above]{$\left[\begin{pmatrix} g_2\\p \end{pmatrix}\right]$} (a-2-2);
\path[->,font=\scriptsize](a-2-2) edge node[above]{$\left[\iota^{1,0} \right]$} (a-2-3);
\path[->,font=\scriptsize](a-2-3) edge node[above]{$\left[\pi^{2} \right]$} (a-2-4);

\path[->,font=\scriptsize](a-1-2) edge node[right]{$[\iota^{0}]$} ($(a-2-2.north) +(0,.3)$);
\path[->,font=\scriptsize]($(a-2-2.south) +(0,-.3)$) edge node[right]{$[\pi^{1}]$} (a-3-2);
\path[->,font=\scriptsize](a-3-2) edge node[right]{$[g_1]$} (a-4-2);
\end{tikzpicture}
\end{equation}
In the diagram, $\iota^{0}$ is the inclusion $C^{0} \hookrightarrow C^{1} \oplus C^{0}$, while $\pi^{1}$ is the projection operator $C^{1} \oplus C^{0} \to C^{1}$. 
Similarly, $\iota^{1,0}$ is the inclusion $C^{1} \oplus C^{0} \hookrightarrow C^{2}\oplus C^{1} \oplus C^{0}$ and $\pi^{2}$ is the projection operator $C^{2} \oplus C^{1} \oplus C^{0} \to C^{2}$. 
By Lemma \eqref{lem:exact-seq}, we see that the row sequence and column sequences are exact.
This implies that
\begin{equation}
    H_1(\partial) = \im [\iota^{1,0}] \cong H_1\begin{pmatrix} \partial^{1} & \\ g_1 & \partial^{0} \end{pmatrix}/ \im \left[\begin{pmatrix} g_2\\p \end{pmatrix}\right]
\end{equation}
where is isomorphism is induced by $[\iota^{1,0}]$.
Also note that $[\pi^{1}]$ induces the following isomorphism
\begin{align}
    H_1\begin{pmatrix} \partial^{1} & \\ g_1 & \partial^{0} \end{pmatrix} \cong \im [\pi^{1}] = \ker [g_1]
\end{align}
In particular, we see that $[\pi^{1}]$ further induces the isomorphism
\begin{align}
    H_1\begin{pmatrix} \partial^{1} & \\ g_1 & \partial^{0} \end{pmatrix}/ \im \left[\begin{pmatrix} g_2\\p \end{pmatrix}\right] &\cong [\pi^{1}] H_1\begin{pmatrix} \partial^{1} & \\ g_1 & \partial^{0} \end{pmatrix}/ [\pi^{1}] \im \left[\begin{pmatrix} g_2\\p \end{pmatrix}\right] \\
    &= \ker[g_1]/\im[g_2] =H_1(\partial^{\bg})
\end{align}
Carefully tracking the previous isomorphisms provides the isomorphism $H_1(\partial^{\bg}) \to H_1 (\partial)$ and thus the statement follows.
More specifically, let $\llb \ell^{1}_{1} \rrb \in H_1(\partial^{\bg})$ and choose $[\ell^{1}_{1}]\in \ker[g_1] \subseteq H_1(\partial^{1})$ to be a representation with further sub-representation $\ell^{1}_{1} \in \ker \partial^{1} \subseteq C^{1}_{1}$.
Since $[g_1\ell^{1}_{1}] = [g_1][\ell^{1}_{1}] =0$ where $[\cdots]$ denotes the corresponding quotient map, we see that $g_1 \ell^{1}_{1} \in \im \partial^{0}$ and thus there exists $\ell^{0}_{1}$ such that $g_1 \ell^{1}_{1} =\partial^{0} \ell^{0}_{1}$.
Further note that
\begin{equation}
    \begin{pmatrix} 
    \partial^{1} & \\ g_1 & \partial^{0} 
    \end{pmatrix}
    \begin{pmatrix}
    \ell^{1}_{1} \\ \ell^{0}_{1} 
    \end{pmatrix} = 0
\end{equation}
where we used the fact that $\partial^{1} \ell^{1}_{1} =0$. 
Hence, we see that the corresponding equivalence class belongs to
\begin{equation}
    \left[ \begin{pmatrix}
    \ell^{1}_{1} \\ \ell^{0}_{1} 
    \end{pmatrix} \right]\in H_1\begin{pmatrix} \partial^{1} & \\ g_1 & \partial^{0} \end{pmatrix}
\end{equation}
And projects to $[\ell^{1}_{1}]$ via $[\pi^{1}]$. Combined with the inclusion map $[\iota^{1,0}]$, we see that the isomorphism $H_1(\partial^{\bg}) \to H_1(\partial)$ is given by
\begin{equation}
    \llb \ell^{1}_{1} \rrb \mapsto \left[\begin{pmatrix} 0 \\ \ell^{1}_{1} \\ \ell^{0}_{1}\end{pmatrix} \right]
\end{equation}
\end{proof}

\begin{proof}[Proof of Theorem \eqref{thm:height-n}]
\label{proof:height-n}
We shall first prove the statement for the usual height-1 cone, and then apply induction to show that the statement holds for all height-$n$ cones for $n> 1$.
Indeed, in the case of a height-1 cone, first consider the case where $m=0$. By Lemma \eqref{lem:exact-seq}, we see that the following is an exact sequence

\begin{equation}
\begin{tikzpicture}[baseline]
\matrix(a)[matrix of math nodes, nodes in empty cells, nodes={minimum size=5mm},
row sep=3em, column sep=3em,
text height=1.5ex, text depth=0.25ex]
{
H_1 (\partial^{1}) &
H_0 (\partial^{0}) & 
H_0
\begin{pmatrix} 
\partial^{1} & \\ 
g_1 & \partial^{0} 
\end{pmatrix} & 
H_0(\partial^1) = 0  \\};
\path[->,font=\scriptsize]
(a-1-1) edge node[above]{$[g_1]$} (a-1-2)
(a-1-2) edge node[above]{$[\iota^{0}]$} (a-1-3)
(a-1-3) edge node[above]{$[\pi^{1}]$} (a-1-4);
\end{tikzpicture}
\end{equation}
where $\iota^{0}$ is the inclusion map $C^{0} \hookrightarrow C^{1} \oplus C^{0}$ and $\pi^{1}$ is the projection map $C^{1} \oplus C^{0} \to C^{1}$.
By definition, we then see that $[\iota^{0}]$ induces an isomorphism from $\coker[g_1] \subseteq H_0(\partial^0)$ to $H_0(\partial^{\cone{(g_1)}})$ and thus the claim holds for $m=0,n=1$. 
Similarly, consider the case where $m=1$ so that by Lemma \eqref{lem:exact-seq}, the following is an exact sequence
\begin{equation}
\begin{tikzpicture}[baseline]
\matrix(a)[matrix of math nodes, nodes in empty cells, nodes={minimum size=5mm},
row sep=3em, column sep=3em,
text height=1.5ex, text depth=0.25ex]
{
H_1 (\partial^{0}) =0 &
H_1
\begin{pmatrix} 
\partial^{1} & \\ 
g_1 & \partial^{0} 
\end{pmatrix} & 
H_1(\partial^1) &
H_0(\partial^{0})
\\};
\path[->,font=\scriptsize]
(a-1-1) edge node[above]{$[\iota^{0}]$} (a-1-2)
(a-1-2) edge node[above]{$[\pi^{1}]$} (a-1-3)
(a-1-3) edge node[above]{$[g_1]$} (a-1-4);
\end{tikzpicture}
\end{equation}
By definition, we then see that $[\pi^{1}]$ induces an isomorphism from $H_1(\partial^{\cone{(g_1)}})$ onto $\ker[g_1] \subseteq H_1(\partial^1)$ and thus the claim holds for $n=1$.

We now prove the statement using induction on $n$, i.e., assume that the statement holds for height $(n-1)$ cones where $n > 1$, and consider a height-$n$ cone.
Let $\partial =\partial^{[0,n]}$ so that
\begin{equation}
    \partial^{[0,n]} = 
    \begin{pmatrix}
    \partial^{(m,n]} &  &&\\
    \begin{pmatrix} \cdot & g_{m+1}\end{pmatrix} & \partial^{m} && \\
    \cdot & \begin{pmatrix} g_{m} \\ \cdot \end{pmatrix} & \partial^{[0,m)}
\end{pmatrix}
\end{equation}
where $\begin{pmatrix} \cdot & g_{m+1}\end{pmatrix}$ is a $1\times (n-m)$ row matrix of maps and similarly, $\begin{pmatrix}  g_m & \cdot\end{pmatrix}^T$ is a $m\times 1$ column matrix of maps.
By induction, note that
\begin{equation}
    H_{m+1} (\partial^{(m,n]}) \cong \coker [g_{m+2}] \subseteq H_{m+1}(\partial^{m+1})
\end{equation}
where $[g_{m+2}]:H_{m+2}(\partial^{m+2}) \to H_{m+1}(\partial^{m+1})$ is that induced by $g_{m+2}$. 
In the case where $m+2>n$ ($m+1>n$), the convention $C^{m+2} =0$ ($C^{m+1}=0$) is used so that $g_{m+2}$ is the zero map.
Equivalently, the following is an exact sequence
\begin{equation}
    H_{m+2}(\partial^{m+2}) \xrightarrow{[g_{m+2}]} H_{m+1}(\partial^{m+1}) \xrightarrow{[\iota^{m+1}]} H_{m+1}(\partial^{(m,n]}) \to 0
\end{equation}
where $\iota^{m+1}$ is the inclusion map $C^{m+1} \hookrightarrow C^{n} \oplus \cdots C^{m+1}$.
Similarly, note that
\begin{equation}
    H_{m-1} (\partial^{[0,m)}) \cong \ker [g_{m-1}] \subseteq H_{m-1}(\partial^{m-1})
\end{equation}
where $[g_{m-1}]:H_{m-1}(\partial^{m-1}) \to H_{m-2}(\partial^{m-2})$ is that induced by $g_{m-1}$ and an analogous convention is used if $m-1<0$ or $m-2 < 0$.
Equivalently, the following is an exact sequence
\begin{equation}
    0 \to H_{m-1}(\partial^{[0,m)}) \xrightarrow{[\pi^{m-1}]} H_{m-1}(\partial^{m-1}) \xrightarrow{[g_{m-1}]} H_{m-2}(\partial^{m-2})
\end{equation}
where $\pi^{m-1}$ is the projection operator $C^{m-1} \oplus \cdots C^{0} \to C^{m-1}$.
Similarly, induction also implies
\begin{equation}
    H_m (\partial^{(m,n]}) = H_m (\partial^{[0,m)}) = 0
\end{equation}
Hence, by Theorem \eqref{thm:height-2}, we see that
\begin{align}
    H_m(\partial) &\cong \ker \left[\begin{pmatrix} g_{m} \\ \cdot \end{pmatrix}\right]/\im [\begin{pmatrix} \cdot & g_{m+1}\end{pmatrix}]\\
    &= \ker \left( [\pi^{m-1}]\left[\begin{pmatrix} g_m \\ \cdot \end{pmatrix}\right]\right)/\im \left([\begin{pmatrix} \cdot & g_{m+1}\end{pmatrix}] [\iota^{m+1}] \right) \\
    &= \ker [g_m]/\im [g_{m+1}] \\
    &= H_m(\partial^{\bg})
\end{align}
where we also used the previous exact sequences.
\end{proof}

\begin{proof}[Proof of Cleaning Lemma \eqref{lem:cleaning}]
\label{proof:cleaning}
By definition (see Example \eqref{ex:CSS}), 
\begin{equation}
    d_Z = \min_{\ell_1\in C_1:0\ne \llb \ell_1^1\rrb \in H_1(\partial)} |\ell_1|
\end{equation}
By Theorem \eqref{thm:height-2}, if $\ell_1$ is a representation of some $\Gamma_1 \in H_1(\partial)\backslash 0$, then it must be equivalent to $\ell_1^1 \oplus \ell_1^0$ with respect to $\partial$ where $\ell_1^1$ is a sub-representation of $\gamma_1 \in H_1(C^{\bg})\backslash 0$, $\ell_1^0$ is such that $g_1\ell_1^1 +\partial^0\ell_1^0=0$ and $\Gamma_1 \mapsto \gamma_1$ via the isomorphism.
In particular, we see that there exists $\ell_2\in C_2$ such that
\begin{equation}
    \ell_1 =
    \begin{pmatrix} 
    0\\
    \ell^{1}_{1} \\
    \ell^{0}_{1} 
    \end{pmatrix} +\partial \ell_{2} = 
    \begin{pmatrix} 
    0\\
    \ell^{1}_{1} \\
    \ell^{0}_{1} 
    \end{pmatrix}
    +
    \begin{pmatrix} 
    \partial^{2} \ell^{2}_{2}\\
    g_2\ell^{2}_{2} +\partial^{1}\ell^{1}_{2}\\
    p\ell^{2}_{2} +g_1\ell^{1}_{2} +\partial^{0}\ell^{0}_{2} 
    \end{pmatrix}
\end{equation}
Hence, we have
\begin{equation}
    d_Z \ge \min_{\ell_1^1\in C_1^1: 0\ne \llb \ell_1^1\rrb \in H_1(\partial^{\bg})} \min_{\ell_2^2,\ell_2^1}(|\partial^2\ell_2^2| + |\ell_1^1+\partial^1\ell_2^1  +g_2\ell_2^2|)
\end{equation}
Note that $[\ell_1^1 +\partial^1 \ell_2^1] = [\ell_1^1]$ in $C^1$ and thus we have
\begin{align}
    d_Z &\ge \min_{\ell_1^1\in C_1^1: 0\ne \llb \ell_1^1\rrb \in H_1(\partial^{\bg})} \min_{\ell_2^2} (|\partial^2\ell_2^2| +|\ell_1^1  +g_2\ell_2^2|)\\
    &\ge \alpha \min_{\ell_1^1\in C_1^1: 0\ne \llb \ell_1^1\rrb \in H_1(\partial^{\bg})} \min_{\ell_2^2} \left(\frac{1}{\alpha} |\partial^2\ell_2^2| +|\ell_1^1  +g_2\ell_2^2| \right)
\end{align}
By the isoperimetric inequality \eqref{eq:isoperimetric}, we see that for any $\ell_2^2$, there exists $\hat{\ell}_2^2$ such that $[\ell_2^2 - \hat{\ell}_2^2] \in H_2(\partial^2)$ and thus $[g_2(\ell_2^2 - \hat{\ell}_2^2)]\in H_1(\partial^1)$. 
Hence, by the triangle inequality,
\begin{equation}
    d_Z \ge \alpha \min_{\ell_1^1\in C_1^1: 0\ne \llb \ell_1^1\rrb \in H_1(\partial^{\bg})} \min_{\ell_2^2} |\ell_1^1+g_2(\ell_2^2 - \hat{\ell}_2^2)| 
\end{equation}
Since $\llb \ell_1^1\rrb =\llb \ell_1^1 +g_2(\ell_2^2-\hat{\ell}_2^2)\rrb$ in $H_1(\partial^{\bg})$, we see that the statement follows, i.e.,
\begin{align}
    d_Z &\ge \alpha \min_{\ell_1^1\in C_1^1: 0\ne \llb \ell_1^1\rrb \in H_1(\partial^{\bg})} \min_{\ell_2^2} |\ell_1^1| \\
    &= \alpha \min_{\ell_1^1\in C_1^1: 0\ne \llb \ell_1^1\rrb \in H_1(\partial^{\bg})}  |\ell_1^1|
\end{align}
In the particular case where $C^{2} =\bigoplus_a C^{2,a}$ is the direct sum of complexes, we see that
\begin{equation}
    |\partial^2 \ell_2^2| =\sum_{a} |\partial^{2,a} \ell^{2,a}_2|
\end{equation}
where $\ell_2^{2,a}$ is the projection of $\ell^2_2$ onto the component $C^{2,a}_2$. 
Hence, one can apply the triangle inequality and the isoperimetric inequality as before so that the statement follows.
\end{proof}

\section{Proof of Topological Codes}

\subsection{Proof of 2D Toric Code}
\begin{proof}[Proof of Theorem \eqref{thm:toric-honeycomb}]
    \label{proof:toric-honeycomb}
    By Fig. \ref{fig:toric-honey}, it's straightforward to check that $C$ (with levels $C^{2},C^{1},C^{0}$) is the toric code on the honeycomb lattice and thus a complex.
    One may also verify that $C$ is a complex directly by expanding the gluing procedure as follows (omitting 0s and the $\otimes$ sign for simplicity)
    \begin{equation}
    \label{eq:honeycomb-diagram}
    \begin{tikzpicture}[baseline]
    \matrix(a)[matrix of math nodes, nodes in empty cells, nodes={minimum size=20pt},
    row sep=2em, column sep=4em,
    text height=1.5ex, text depth=0.25ex]
    {
    & X_1Y_1\\
    &  X_1 Y_0 \oplus X_0 Y_1\\
    X_0 Y_0  R_1(2) &  X_0 Y_0 R_0(2) \\};
    \path[->,font=\scriptsize]
    (a-1-2) edge node[right]{$g_2$} (a-2-2)
    (a-2-2) edge node[right]{$g_1$} (a-3-2)
    (a-3-1) edge node[above]{$\partial^0 \equiv \partial^{R(2)}$}(a-3-2)
    (a-1-2) edge[bend right,dashed] node[left]{$p$} (a-3-1);
    \end{tikzpicture}
    \end{equation}
    Since the diagram has no squares (excluding the dashed line), it is trivially a commutative diagram.
    It's straightforward to check that $g_1 g_2 = \partial^0 p$ and thus $C$ is a complex.
    Note that trivially, $H_2(C^{2}) = \Tor_{2}^{\sq}$ and $H_1(C^{1}) = \Tor_{1}^{\sq}$ and thus the induced mapping by $g_2$ on the embedded code is merely $\partial_2^{\sq}$.
    By Lemma \eqref{lem:rep}, we note that $H_0(C^{0}) \cong \Tor_{0}^{\sq}$ and thus we can label the basis elements of $H_0(C^{0})$ via $[\|x_0 y_0\ket] \equiv [|x_0 y_0,i\ket]$ for any $i=1,2$.
    Moreover, note that
    \begin{align}
        [g_1] [\|x_1 y_0\ket] &= [|x_1^+ y_0 2\ket +|x_1^- y_0 1\ket] = [\|x_1^+ y_0 \ket] +[\|x_1^- y_0 \ket] \\
        [g_1] [\|x_0 y_1\ket] &= [|x_0 y_1^+ 1\ket +|x_0 y_1^- 2\ket] = [\|x_0 y_1^+ \ket] +[\|x_0 y_1^- \ket]
    \end{align}
    Hence, $[g_1] = \partial^{\sq}_{1}$ and thus $\Tor^{\sq}$ is the embedded code in $C$
\end{proof}

\begin{proof}[Proof of Theorem \eqref{thm:toric-triangular}]
    \label{proof:toric-triangular}
    By Fig. \ref{fig:toric-triangle}, it's straightforward to check that $C$ (with levels $C^{2},C^{1},C^{0}$) is the toric code on the honeycomb lattice and thus a complex.
    The statement follows from Theorem \eqref{thm:height-2} and Lemma \eqref{lem:rep}. and expanding the gluing procedure as follows (omitting 0s and the $\otimes$ sign for simplicity)
    \begin{equation}
    \label{eq:triangle-diagram}
    \begin{tikzpicture}[baseline]
    \matrix(a)[matrix of math nodes, nodes in empty cells, nodes={minimum size=20pt},
    row sep=2em, column sep=4em,
    text height=1.5ex, text depth=0.25ex]
    {
    X_1Y_1 R_0(2) & X_1 Y_1 R_1(2)\\
    X_1 Y_0 \oplus X_0 Y_1 &\\
    X_0 Y_0   &   \\};
    \path[->,font=\scriptsize]
    (a-1-1) edge node[right]{$g_2$} (a-2-1)
    (a-2-1) edge node[right]{$g_1$} (a-3-1)
    (a-1-1) edge node[above]{$\partial^2 \equiv (\partial^{R(2)})^T$}(a-1-2)
    (a-1-2) edge[bend left,dashed] node[below]{$p$} (a-3-1);
    \end{tikzpicture}
    \end{equation}
\end{proof}

\subsection{Proof of Barycentric Subdivision}

\begin{proof}[Proof of Theorem \eqref{thm:barycentric-subdivision}]
    \label{proof:barycentric-subdivision}
    We shall prove by induction.
    Indeed, if $\Delta$ is a finite simplicial complex of length 0, the statement is trivially true, and thus we shall assume that the statement is true for simplicial complexes of length $n-1$ where $n\ge 1$, and consider the simplicial complex $\Delta$ of length $n$.
    Given an $s$-simplex $\| s\ket$ of $\Delta$ with $s\le n$, let $\Delta^{\| s\ket}$ be the induced simplicial subcomplex with differential $\partial^{\| s\ket}$ and let 
    \begin{equation}
        \Delta^{\partial\| s\ket} = \Delta_{s-1}^{\| s\ket} \to \cdots \to \Delta_{0}^{\| s\ket}
    \end{equation}
    be the further subcomplex, which represents the simplicial complex induced by the \textit{boundary} of the $s$-simplex.
    Then it's straightforward to check that
    \begin{equation}
        C^{s} \cong  \bigoplus_{\| s\ket} (C(\Delta^{\partial \| s\ket}) \to \bra \varnothing\ket)
    \end{equation}
    where $C(\Delta^{\partial \| s\ket})$ is the barycentric subdivison of $\Delta^{\partial \| s\ket}$, augmented with the map that maps a $0$-simplex of $C(\Delta^{\partial \| s\ket})$ to the unique basis element of $\bra \varnothing\ket = \dF_2$.
    (Since the differential map of simplicial complexes removes a vertex from a simplex, the augmented differential map follows this reasoning and maps any vertex to the empty set $\varnothing$).
    By induction, we see that $C(\Delta^{\partial\| s\ket})$ and $\Delta^{\partial\| s\ket}$ are quasi-isomorphic.
    Since $C(\Delta^{\partial \|s\ket}), \Delta^{\partial \|s\ket}$ have the same 0-cells, the augment complexes $C(\Delta^{\partial\| s\ket}) \to \bra \varnothing\ket$ and $\Delta^{\partial\| s\ket} \to \bra \varnothing\ket$ are quasi-isomorphic, and thus we only need to consider the homologies of the latter\footnote{The homologies of the augmented complex are often referred as the \textbf{reduced homologies} \cite{hatcher2005algebraic}}.
    In fact, since $\Delta^{\| s\ket}$ only has one $s$-simplex, it's sufficient to show that the the homology of the augmented complex $\Delta^{\| s\ket} \to \bra \varnothing\ket$, is zero, that is, $\tilde{H}_m (\Delta^{\| s\ket}) \equiv H_m(\Delta^{\| s\ket} \to \bra \varnothing\ket)=0$ for all $m$.

    Let $|s\ket,...|0\ket$ be vertices of the $s$-simplex $\| s\ket$ and consider the map $\pi_i:\Delta^{\| s\ket}_{i} \to \Delta^{\| s\ket}_{i}$ such that $\pi_i =0$ for $i> 0$ and $\pi_0$ maps all vertices of the $s$-simplex to $|s\ket$.
    It's clear that $\pi$ is a chain map on the simplicial complex $\Delta^{\| s\ket}$.
    Now define $p$ as the map which maps any $m$-simplex $|\bm{\ell}\ket \equiv |\ell_{m}>\cdots >\ell_{0}\ket$ of $\Delta^{\| s\ket}$ to $|s,\bm{\ell}\ket \equiv |s,\ell_{m},...,\ell_{0}\ket$ if $s>\ell_{m}$ and to $0$ otherwise. 
    Then it's straightforward to check that $I+\pi = \partial^{\| s\ket} p + p\partial^{\| s\ket}$ so that $\pi$ is chain-homotopic to the identity.
    Therefore, the reduced homology must always be zero.
    Finally, one can check that
    \begin{equation}
        \partial^{s} [\| s\ket] =0, \quad [\| s\ket]\equiv \sum_{\|m\ket,m<s: \| s\ket \sim^{\Delta} \| s-1\ket \sim^{\Delta} \cdots \sim^{A} \| 0\ket}\|s,s-1,...,0\ket
    \end{equation}
    by utilizing the fact that if $\|m+1\ket \sim^{\Delta} \|m-1\ket$ are given, then there are exactly two $m$-simplices such that $\|m+1\ket \sim^{\Delta} \|m\ket \sim^{\Delta} \|m-1\ket$
\end{proof}

\section{Proof of Euclidean Embedding}
\subsection{Proof of Layer Code}

\begin{proof}[Proof of Theorem \eqref{thm:layer-code}]
\label{proof:layer-code}
Our goal is to check that the lower triangular matrix defined in Theorem \eqref{thm:height-2} satisfies $\partial \partial =0$.
It's then sufficient to show that (see Diagram \eqref{eq:height-2-diagram})
\begin{align}
    g_2 \partial^{2} &= \partial^{1} g_2 \\
    g_1 \partial^{1} &= \partial^{0} g_1 \\
    g_1 g_2 &=  \partial^{0} p+ p\partial^{2}
\end{align}
Where the last equality implies that $p$ is a chain-homotopy and $g_2 g_1$ is chain-homotopic to 0.
Indeed, note that
\begin{align}
    g_2\partial^{2} |y_0 z_0;x_0\ket &= g_2 \left[\sum_{|y_1\ket \sim^{Y} |y_0\ket }|y_1 z_0;x_0\ket + \sum_{|z_1\ket\sim^{Z} |x_0\ket} |y_0 z_1;x_0\ket \right]\\
    &= \sum_{|z_1\ket\sim^{Z} |z_0\ket} 1\{|y_0\ket \sim^{A} |x_0\ket \} |x_0 z_1;y_0\ket \\
    &=\partial^{1} \left[ 1\{|y_0\ket \sim^{A} |x_0\ket\} |x_0 z_0;y_0\ket \right] \\
    &= \partial^{1} g_2|y_0 z_0;x_0\ket 
\end{align}
Similarly, note that
\begin{align}
    g_1 \partial^{1} |x_1 z_0;y_0\ket &= g_1 \left(\sum_{|x_0\ket \sim^{X} |x_1\ket} |x_0 z_0;y_0\ket +\sum_{|z_1\ket \sim^{Z} |z_0\ket} |x_1 z_1;y_0\ket\right) \\
    &= \sum_{|x_0\ket \sim^{X} |x_1\ket} 1\{|y_0\ket \sim^{A} |z_0\ket\}|x_0 y_0 ;z_0\ket \\
    &= \partial^{0}  |x_1 y_0;z_0\ket 1\{|y_0\ket \sim^{A} |z_0\ket\}  \\
    &= \partial^{0} g_1 |x_1 z_0;y_0\ket
\end{align}
Finally, let us show that $p$ is a chain-homotopy with respect to $g_1 g_2$ and $0$.
Indeed, note that
\begin{align}
    \partial^{0} p |y_0 z_0;x_0\ket &= \partial^{0} \left(1\{|y_0\ket \in Y[|x_0\ket\wedge |z_0\ket)\} |x_0 y_0^+;z_0\ket \right) \\
    &= 1\{|y_0\ket \in Y[|x_0\ket\wedge |z_0)\} \left( |x_0 y_0;z_0\ket +|x_0 (y_0 +1);z_0\ket \right)
\end{align}
And that
\begin{align}
    p\partial^{2} |y_0 z_0;x_0\ket &= p\left(\sum_{|y_1\ket \sim^{Y} |y_0\ket }|y_1 z_0;x_0\ket + \sum_{|z_1\ket\sim^{Z} |x_0\ket} |y_0 z_1;x_0\ket \right) \\
    &= p \left( |y_0^- z_0;x_0\ket +|y_0^+ z_0;x_0\ket\right) \\
    &= 1\{|y_0^-\ket \in Y[|x_0\ket\wedge |z_0)\} |x_0 y_0;z_0\ket + 1\{|y_0^+\ket \in Y[|x_0\ket\wedge |z_0)\} |x_0 (y_0+1) ;z_0\ket
\end{align}
Note that if $|y_0\ket$ is not a start or end of a string defect, then $|y_0^\pm\ket$ are either both in or both not in the string defect. 
Conversely, if $|y_0\ket \in Y[|x_0\ket \wedge |z_0\ket)$ is a start or end of a string defect, then one and only of $|y_0^\pm\ket$ is in the string defect. 
Therefore,
\begin{align}
    (\partial^{0} p +p\partial^{2})|y_0 z_0;x_0\ket &= 1\{|x_0\ket \sim^{A} |y_0\ket \sim^{A} |z_0\ket \} |x_0 y_0;z_0\ket \\
    &= g_1 g_2 |y_0 z_0;x_0\ket
\end{align}
Hence, $C$ is a complex, i.e., $\partial \partial =0$.

Next, we need to show that the embedded CSS code $C^{\bg}$ induced by quotients $[g_2],[g_1]$ is indeed the original CSS code $A$.
By Lemma \eqref{lem:rep} and the K\"unneth formula \eqref{lem:Kunneth}, we see that $C$ is a regular cone and that
\begin{enumerate}[label=\arabic*)]
    \item For fixed $x_0$, $H_2(Y^T\otimes Z^T) \cong \dF_2$ has unique basis element
    \begin{equation}
        \label{eq:layer-code-H2}
        \left[\|x_0\ket\right], \quad \|x_0\ket \equiv\sum_{y_0,z_0} |y_0 z_0;x_0\ket
    \end{equation}
    so that $[\|x_0\ket],x_0=1,...,n_Z$ form a basis for $H_2(\partial^{2})$.
    \item For fixed $y_0$, $H_1(X\otimes Z^T) \cong \dF_2$ has unique basis element
    \begin{equation}
        \label{eq:layer-code-H1}
        \left[ \|y_0\ket \right], \quad \|y_0\ket \equiv \sum_{z_0} |x_0 z_0;y_0\ket
    \end{equation}
    where $x_0$ can be chosen arbitrarily, so that $[\|y_0\ket],y_0=1,...,n_Z$ form a basis for $H_1(\partial^{1})$.
    Note the summation is over $z_0$, which can be thought of a string operator in the $z$-direction (compare with Fig. \ref{fig:toric-alter}).
    \item For fixed $z_0$, $H_0(X\otimes Y) \cong \dF_2$ has unique basis element, which can be represented by $\|z_0\ket \equiv |x_0 y_0;z_0\ket$ where we have chosen $x_0,y_0$ arbitrarily, and thus $[\|z_0\ket]$ form a basis for $H_0(\partial^0)$
\end{enumerate}
Hence, we see that
\begin{align}
    [g_2][\|x_0\ket] &= [g_2] \left[\sum_{y_0,z_0} |y_0 z_0;x_0\ket\right] \\
    &= \left[g_2 \sum_{y_0,z_0} |y_0 z_0;x_0\ket \right] \\
    &= \sum_{y_0} 1\{|y_0\ket \sim^{A} |x_0\ket\} \left[\sum_{z_0} |x_0 z_0; y_0 \ket\right] \\
    & = \sum_{y_0} 1\{|y_0\ket \sim^{A} |x_0\ket\} [\|y_0\ket]
\end{align}
The case is similar for $[g_1]$ and the embedded code $C^{\bg}$ is exactly the original CSS code $A$
\end{proof}

\begin{proof}[Proof of Theorem \eqref{thm:layer-code-distance}]
\label{proof:layer-code-distance}
We shall prove the lower bound for $d_Z$ --  the case for $d_X$ is similar and thus omitted. 
Indeed, following the proof of the Cleaning Lemma \eqref{lem:cleaning}, we have
\begin{align}
    d_Z &\ge \frac{2}{w_Z} \min_{\ell_1^1\in C_1^1: \llb \ell_1^1\rrb \in H_1(A)\backslash 0} \min_{\ell_2^2} \left(\frac{w_Z}{2} |\partial^2\ell_2^2| +|\ell_1^1  +g_2\ell_2^2| \right)
\end{align}
Let $\ell_2^2(x_0 z_0)$ denote the collection of $|y_0 z_0;x_0\ket \in \ell_2^2$ for given $x_0,z_0$, i.e., the $(x_0,z_0)$-projection of $\ell_2^2$, and similarly define $\ell_1^1(x_0 z_0)$.
If $\delta^{Y} =(\partial^{Y})^T$ is the codifferential of $Y$, then we have
\begin{equation}
    \frac{w_z}{2} |\partial^2\ell_2^2| +|\ell_1^1  +g_2\ell_2^2| \ge \sum_{x_0,z_0} \left(\frac{w_Z}{2} |\delta^{Y} \ell_2^2(x_0 z_0)| + |\ell_1^1 (x_0 z_0) +g_2 \ell_2^2(x_0z_0)| \right)
\end{equation}
Note that by definition of $g_2$, we have $|g_2 \ell_2^2(x_0z_0)| \le w_Z$. In fact, either $|g_2 \ell_2^2(x_0z_0)|$ or $|g_2 (\ell_2^2(x_0 z_0) + \dnum (x_0z_0))| \le w_Z/2$ where
\begin{equation}
    \dnum (x_0z_0) \equiv \sum_{y_0} |y_0 z_0;x_0\ket 
\end{equation}
Note that $\dnum(x_0,z_0)$ is the $z_0$-projection of $\|x_0\ket$ defined in Eq. \eqref{eq:layer-code-H2}.
Note that $\ell_2^2(x_0 z_0) \subseteq  \dnum (x_0z_0)$ as subsets and thus $\ell_2^2(x_0 z_0) + \dnum (x_0z_0)$ can be regarded as the \textit{complement} of $\ell_2^2(x_0 z_0)$.
Note that the boundary term is invariant under complement, i.e.,
\begin{equation}
    |\delta^{Y} \ell_2^2(x_0z_0)| =|\delta^{Y} (\ell_2^2(x_0z_0) +\dnum(x_0 z_0))| 
\end{equation}
By the triangular inequality, we have
\begin{equation}
    \frac{w_Z}{2} |\delta^{Y} \ell_2^2(x_0z_0)| + |\ell_1^1(x_0z_0) +g_2 \ell_2^2(x_0z_0)| \ge \min\left(|\ell_1^1(x_0z_0)|, |\ell_1^1(x_0z_0) + g_2 \dnum(x_0z_0)| \right)
\end{equation}
Therefore, we have
\begin{equation}
    d_Z \ge \frac{2}{w_Z} \min_{\ell_1^1\in C_1^1: \llb \ell_1^1\rrb \in H_1(A)\backslash 0} \sum_{x_0,z_0} \min\left(|\ell_1^1(x_0z_0)|, |\ell_1^1(x_0z_0) + g_2 \dnum(x_0z_0)| \right)
\end{equation}
Fix $\ell_1^1$ and $z_0$. We claim that
\begin{equation}
    \sum_{x_0} \min\left(|\ell_1^1(x_0z_0)|, |\ell_1^1(x_0z_0) + g_2 \dnum(x_0z_0)| \right) \ge \min_{\ell_2^A \in A_2} |\ell_1^1(z_0) +g_2 \dnum(\ell_2^A,z_0)|
\end{equation}
where $\dnum(\ell_2^A,z_0)$ is the sum over $\dnum(x_0,z_0)$ where $|x_0\ket \in \ell_2^{A}$.
Indeed, start with $\ell_2^{A}$ as the emptyset. For every $x_0$ on the left-hand-side (LHS), determine whether $|\ell_1^1(x_0z_0)|$ or $|\ell_1^1(x_0z_0)+g_2 \dnum(x_0 z_0)|$ is the smaller value. 
If the former, do nothing; otherwise, add $|x_0\ket$ to $\ell_2^{A}$.
It's then straightforward to check that for the constructed $\ell_2^{A}$ (depending on $\ell_1^1,z_0$), we have
\begin{equation}
    \sum_{x_0} \min\left(|\ell_1^1(x_0z_0)|, |\ell_1^1(x_0z_0) + g_2 \dnum(x_0z_0)| \right)  \ge |\ell_1^1(z_0) +g_2 \dnum(\ell_2^A,z_0)|
\end{equation}
And thus the claim follows.
Using the fact that $g_2 \dnum(\ell_2^A) \in \ker \partial^1$ and induces the same equivalence class in $H_1(A)$ as $\ell_1^1$, i.e., $\llb \ell_1^1\rrb = \llb \ell_1^1 +g_2 \dnum(\ell_2^A)\rrb$, we have
\begin{align}
    d_Z &\ge \frac{2}{w_Z} \min_{\ell_1^1\in C_1^1:\llb \ell_1^1\rrb \in H_1(A)\backslash 0} \sum_{z_0} \min_{\ell_2^A \in A_2} |\ell_1^1(z_0) +g_2 \dnum(\ell_2^A,z_0)|\\
    &\ge \frac{2}{w_Z}\sum_{z_0} \min_{\ell_1^1\in C_1^1:\llb \ell_1^1\rrb \in H_1(A)\backslash 0} \min_{\ell_2^A \in A_2} |\ell_1^1(z_0) +g_2 \dnum(\ell_2^A,z_0)| \\
    &\ge \frac{2}{w_Z}\sum_{z_0} \min_{\ell_1^1\in C_1^1:\llb \ell_1^1\rrb \in H_1(A)\backslash 0} |\ell_1^1(z_0)|
\end{align}
Note that if $|\ell_1^1(z_0)| \ge |[\ell_1^1]|$ for all $z_0$ where $|[\ell_1^1]|$ is the weight of $[\ell_1^1]$ treated as an element in $A_1$, the statement then follows and thus it's sufficient to prove the inequality. Indeed, given $\ell_1^1$ with $\ell_1^A = [\ell_1^1]$, we see that there exists $\ell_2^1 \in C_2^1$ such that
\begin{equation}
    \ell_1^1 = \sum_{y_0} 1\{|y_0\ket \in \ell_1^A\} \|y_0\ket+\partial^1 \ell_2^1
\end{equation}
where $\|y_0\ket$ is defined in Eq. \eqref{eq:layer-code-H1}. Hence,
\begin{align}
    \ell_1^1(z_0) &= \sum_{y_0} 1\{|y_0\ket \in \ell_1^A\} |x_0 z_0;y_0\ket + \partial^1 \sum_{y_0,x_1} |x_1 z_0;y_0\ket 1\{|x_1 z_0;y_0\ket \in \ell_2^1\} \\
    |\ell_1^1(z_0)| &\ge \sum_{y_0} 1\{|y_0\ket \in \ell_1^A\} ||x_0\ket + \partial^{X}\ell_2^1(y_0z_0) | \\
    &\ge |\ell_1^A|
\end{align}
where we treat $\ell_2^1(y_0z_0)$ as a subset of 1-cells $|x_1\ket$ in $X$.
\end{proof}
\twocolumngrid

\end{document}